\documentclass[12pt,draftclsnofoot,onecolumn,twoside,letterpaper]{IEEEtran}

\usepackage[noadjust]{cite}
\usepackage{amsmath,amssymb,amsfonts,amscd,amsthm,amsbsy}
\usepackage{graphicx}
\usepackage[english]{babel}
\usepackage{mathrsfs}
\usepackage{mathtools}
\usepackage{array}
\usepackage{blkarray}
\usepackage{caption}
\usepackage{subcaption}
\newtheorem{theorem}{Theorem}
\newtheorem{lemma}{Lemma}

\newtheorem{corollary}{Corollary}
\theoremstyle{definition}
\newtheorem{remark}{Remark}

\theoremstyle{definition}

\theoremstyle{definition}
\newtheorem{example}{Example}
\newtheorem{definition}{Definition}
\theoremstyle{definition}

\theoremstyle{definition}

\usepackage{tabulary}
\usepackage{makecell}

\usepackage{etoolbox}
\AtEndEnvironment{example}{\null\hfill\IEEEQED}%
\AtEndEnvironment{remark}{\null\hfill\IEEEQED}%

\newcommand{\Fb}{\mathbb{F}}
\newcommand{\Zb}{\mathbb{Z}}

\newcommand{\xb}{{\bf{x}}}
\newcommand{\eb}{{\bf{e}}}
\newcommand{\bxi}{{\pmb{\xi}}}

\newcommand{\Hb}{{\bf{H}}}
\newcommand{\hb}{{\bf{h}}}

\newcommand{\Ib}{{\bf{I}}}

\newcommand{\cb}{{\bf{c}}}
\newcommand{\wb}{{\bf{w}}}

\newcommand{\ei}{\epsilon}
\newcommand{\colspan}{\mathcal{C}}
\newcommand{\rank}{{\rm rank}}
\newcommand{\lopt}{\ell^*}

\newcommand{\Ab}{{\textbf{A}}}
\newcommand{\Bb}{{\textbf{B}}}

\newcommand{\ical}{\mathcal{I}}

\newcommand{\fcal}{\mathcal{F}}
\newcommand{\kcal}{\mathcal{K}}

\newcommand{\acal}{\mathcal{A}}
\newcommand{\ycal}{\mathcal{Y}}

\newcommand{\wtm}{\mathrm{wt}}

\newcommand{\xcal}{\mathcal{X}}

\newcommand{\vand}{{\rm Vand}}

\normalsize

\title{Blind Updates in Coded Caching}

\author{Suman~Ghosh, Prasad~Krishnan and Lakshmi~Prasad~Natarajan

\thanks{This paper was presented in part at the 2020 IEEE Information Theory Workshop (ITW 2020), Riva del Garda, Italy.}%
\thanks{Suman Ghosh and Lakshmi Prasad Natarajan are with the Department of Electrical Engineering, Indian Institute of Technology Hyderabad, Sangareddy 502 285, India (email:
{ee16resch11006, lakshminatarajan}@iith.ac.in).}
\thanks{Prasad Krishnan is with the Signal Processing and Communications Research Center,
International Institute of Information Technology, Hyderabad 500032, India.
(email: prasad.krishnan@iiit.ac.in)}
}

\begin{document}

\maketitle

\begin{abstract}
We consider the centralized coded caching system where a library of files is available at the server and their subfiles are cached at the clients as prescribed by a placement delivery array (PDA). We are interested in the problem where a specific file in the library is replaced with a new file at the server, the contents of which are correlated with the file being replaced, and this change needs to be communicated to the caches. Upon replacement, the server has access only to the updated file and is unaware of its differences with the original, while each cache has access to specific subfiles of the original file as dictated by the PDA. We model the correlation between the two files by assuming that they differ in at the most $\ei$ subfiles, and aim to reduce the number of bits broadcast by the server to update the caches. We design a new elegant coded transmission strategy for the server to update the caches blindly, and also identify a simple scheme that is based on MDS codes. We then derive converse bounds on the minimum communication cost $\lopt$ among all linear strategies. For two well-known families of PDAs -- Maddah-Ali \& Niesen's caching scheme and a PDA by Tang \& Ramamoorthy and Yan et al. -- our new scheme has cost $\lopt(1 + o(1))$ when the updates are sufficiently sparse, while the scheme using MDS codes has order-optimal cost when the updates are dense. 
\end{abstract}


\section{Introduction}

Coded caching is a powerful technique that utilizes memory as a resource to offset communication costs~\cite{Ali_Niesen_2014}. A systematic placement of files in the client caches and a careful design of coded transmissions can significantly reduce the number of bits being broadcast during file delivery.
Coded caching has been shown to provide benefits in a variety of scenarios, such as in wireless networks~\cite{CacheAidedInterference,LampirisElia}, decentralized caching~\cite{DecentralizedPaper}, in D2D networks~\cite{D2Dpaper}, and in the presence of relays~\cite{CombinationNetworks}.
Caching could play a key role in overcoming bandwidth bottlenecks in next generation communication networks.

We consider centralized coded caching schemes that are designed based on \emph{placement delivery arrays} or PDAs~\cite{YCTC_2017}.
A PDA is a structure that provides~\emph{(i)} a scheme for partitioning every file in the library into subfiles and placing these subfiles in the caches, and~\emph{(ii)} a method to generate coded transmissions to meet the demand of the clients during the delivery phase. The notion of PDAs provides a systematic framework for designing centralized coded caching schemes. 
A number of works in the literature have designed coded caching schemes to achieve a graceful trade-off between communication rate and subpacketization level, and several of them -- including Maddah-Ali \& Niesen scheme of~\cite{Ali_Niesen_2014} -- fall into the framework of PDAs, such as~\cite{YCTC_2017,YTCC_COMML_18,SZG_IT_18,TaR_IT_18,Cheng_TComm_19,Cheng_arXiv_17,Chittoor_GC_19}.
Apart from coded caching, PDAs are also known to play an important role in the design of coded distributed computing techniques~\cite{LMYA_IT_17,RaK_ISIT19,KoR_GLOBECOM_18}.

In this paper, we consider the scenario where one of the files in the library is replaced with a new file {of equal size}. 
We assume that the contents of the new file are written over the contents of the file being replaced, and hence, the server loses the old file once the library is updated.
This is a reasonable mode of operation when the server, in order to save memory and reduce system complexity, does not intend to maintain a detailed log of the changes in the library and the store the previous versions of the state of the library. 
Once the file is replaced at the server, these changes must be communicated to the caches in order to update their contents.
We consider the scenario where the contents of the new file and the old file (that has been removed at the server) are correlated, and we design coded transmission strategies for updating the cache contents that exploit this correlation.
We model the correlation between the two in terms of the number of subfiles in which the two files differ.

We assume that the files are cached using a $(K,F,Z,S)$-PDA, where $K$ is the number of caches, $F$ is the number of subfiles that each file is split into, $Z$ is the number of subfiles of each file that is cached at each client, and $S$ denotes the number of transmissions in the delivery phase of coded caching. 
We further impose the conditions (which are satisfied by several PDAs known in the literature) that every subfile must be cached in a constant number of caches, say $r$ caches out of $K$, and each subfile must be cached in a distinct subset of users (see~\eqref{eq:mild_condition_pda}).
If there exist two or more users whose cache contents are identical, without loss of generality, we replace this set of users by a single user with the same cache content.
Under such scenario, suppose that a specific file is replaced with another file, whose contents differ in at the most $\ei$ subfiles. 
We assume that the server knows the value of $\ei$, but since the server has lost the older file, it does not know in which subfiles the two files differ. 
The objective is to design a strategy that can be used by the server to communicate this update to the users.
Since the server is ignorant of the exact difference between the old and new files, we refer to this process as \emph{blind cache update} or simply \emph{cache update}.
We use the number of bits transmitted by the server as the performance metric of an update scheme.

In a naive strategy the server broadcasts the entirety of the new file, which can be directly used by the clients to update their respective caches.
However this strategy does not exploit the side information available at the clients, which are the subfiles of the older file available in their cache.
In this paper, we show that the communication cost can be reduced by using carefully designed coded transmissions that exploit this side information. 
Each user will use the coded transmissions from the server and the subfiles of the older file that are available in its cache to decode the newer version of these subfiles.

\begin{example} \label{ex:toy_example}
\begin{figure}[!t]
\centering
\includegraphics[width=5in]{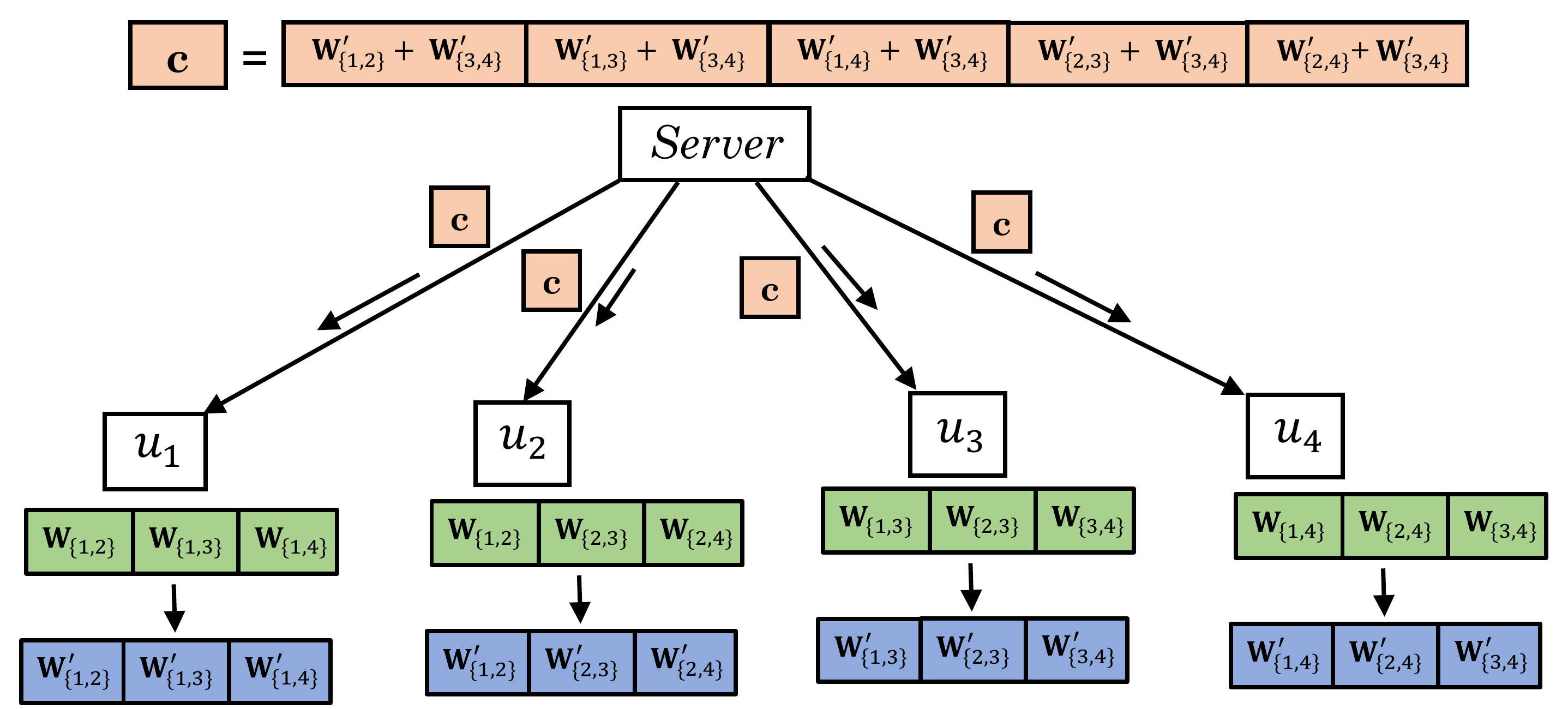}
\caption{The coded blind cache update scheme in Example~\ref{ex:toy_example}.}
\label{fig:toy_example}
\end{figure} 
In this toy example, we illustrate how the older subfiles available in user caches can be used to reduce the communication cost.
Consider the Maddah-Ali \& Niesen's caching scheme across $K=4$ users $u_1,\dots,u_4$ with caching ratio $Z/F=1/2$.
A file ${\bf w}$ is split into $6$ subfiles ${\bf w} = (w_T, T \in \binom{[4]}{2})$, where $\binom{[4]}{2}$ is the collection of all $2$-sized subsets of $[4]=\{1,\dots,4\}$. 
For simplicity, assume that each subfile is a single bit, i.e., $w_T \in \Fb_2$ for all $T \in \binom{[4]}{2}$.
User $u_k$ caches subfile $w_T$ if and only if $k \in T$.
Suppose ${\bf w}$ is updated to a newer version ${\bf w}'=(w'_T, T \in \binom{[4]}{2})$, such that ${\bf w}$ and ${\bf w}'$ differ in at the most one subfile. 
That is, the Hamming weight of the update vector ${\bf e}=(e_T, T \in \binom{[4]}{2}) \triangleq {\bf w}' - {\bf w}$ is at the most $\ei=1$. 
If a naive update scheme is used, the server will broadcast the entire updated file ${\bf w}'$ consisting of $6$ bits. Each user $u_k$ can obtain its new cache content $(w'_T~|~T \in \binom{[4]}{2}, k \in T)$ from the corresponding sub-vectors of ${\bf w}'$.
However, if we use the coded transmission shown in Fig.~\ref{fig:toy_example}, cache updates can be effected using $5$ transmitted bits. 
We illustrate the decoding procedure at $u_1$. The procedure at other caches are similar.

User $u_1$ intends to recover $(w'_{\{1,2\}},w'_{\{1,3\}},w'_{\{1,4\}}) = (w_{\{1,2\}},w_{\{1,3\}},w_{\{1,4\}}) + (e_{\{1,2\}},e_{\{1,3\}},e_{\{1,4\}})$, or equivalently, recover $(e_{\{1,2\}},e_{\{1,3\}},e_{\{1,4\}})$ since he already knows $(w_{\{1,2\}},w_{\{1,3\}},w_{\{1,4\}})$. 
In order to decode $(e_{\{1,2\}},e_{\{1,3\}},e_{\{1,4\}})$, $u_1$ computes a syndrome $(s_1,s_2)$ using the components of the broadcast codeword and its present cache contents as follows 
\begin{align*}
s_1 \triangleq (w'_{\{1,2\}} + w'_{\{3,4\}}) + (w'_{\{1,3\}} + w'_{\{3,4\}}) + w_{\{1,2\}} + w_{\{1,3\}} = e_{\{1,2\}} + e_{\{1,3\}}, \\
s_2 \triangleq (w'_{\{1,3\}} + w'_{\{3,4\}}) + (w'_{\{1,4\}} + w'_{\{3,4\}}) + w_{\{1,3\}} + w_{\{1,4\}} = e_{\{1,3\}} + e_{\{1,4\}}.
\end{align*} 
Note that $(s_1,s_2)$ is the syndrome of $(e_{\{1,2\}},e_{\{1,3\}},e_{\{1,4\}})$ for the length-$3$ binary repetition code. Since the Hamming weight of $(e_{\{1,2\}},e_{\{1,3\}},e_{\{1,4\}})$ is at the most $1$ and since the length-$3$ repetition code is a $1$-error correcting code we conclude that $(s_1,s_2)$ is sufficient to identify $(e_{\{1,2\}},e_{\{1,3\}},e_{\{1,4\}})$.
\end{example}

The coding schemes designed in this paper for blind cache updates are also useful in the following communication scenario known as \emph{broadcasting with noisy side information (BNSI)}~\cite{SuN_2019}.
In the \emph{placement phase} of a coded caching system, the server broadcasts each file in the library to the users, and each user stores a specific subset of the subfiles of each file in its cache.
Now suppose that the channel between the server and the users faces a temporary outage during the transmission of a particular file.
Because of the outage, the users receive erroneous versions of the subfiles of this file and request the server for a retransmission.
In the retransmission phase, each user demands a specific collection of subfiles from the server, and each user has side information in the form of a noisy version of the set of subfiles that it demands from the server.
The task of designing a communication scheme for this retransmission phase is a BNSI problem (we describe this problem formally in Section~\ref{bnsi}).
Our cache update coding strategies serve as solutions to this BNSI problem (details are in Section~\ref{cacheupdate_bnsi}).

\subsection{Contributions}

\begin{table}[!t]
\centering
\begin{tabular}{|l|l|c|c|}
\Xhline{5\arrayrulewidth}
\textsc{Section}   & \textsc{Summary of Contents} & \textsc{Caching Scheme} & \textsc{Main Results} \\
\Xhline{5\arrayrulewidth}
\hline
Sec~\ref{sec:system_model} & System model. Review of PDAs & Any PDA & \\
\hline
\hline
Sec~\ref{cacheupdate_bnsi} & Cache update and BNSI problems are equivalent & Any PDA & Theorem~\ref{thm:cache_update_BNSI} \\
\hline
Sec~\ref{prep} & Design criterion for linear cache update schemes & Any PDA & Theorem~\ref{thrm1}\\
\hline
Sec~\ref{prep} & A cache update scheme based on MDS codes & Any PDA & Lemma~\ref{ub1} \\
\hline
\hline
Sec~\ref{sec:coding_scheme} & A new cache update scheme & PDAs satisfying~\eqref{eq:mild_condition_pda} & Theorem~\ref{thm:new_coding_scheme} \\
\hline
\hline
Sec~\ref{sub:sec:converse:general} & Condition for $\lopt=F$ (for naive update scheme to be optimal) & Any PDA & Lemma~\ref{lem:condition_lopt_F} \\
\hline
Sec~\ref{sub:sec:converse:general} & Proving that $\lopt \geq 2\ei$. Conditions for $\lopt=2\ei,2\ei+1,2\ei+2$ & Any PDA & Lemmas~\ref{lem:cost_2ei},~\ref{lem:cost_2ei_plus} \\
\hline
Sec~\ref{sub:sec:converse:general} &  A generic lower bound for any linear update scheme & Any PDA & Theorem~\ref{thm:generic_lower_bound} \\
\hline
\hline
Sec~\ref{sub:sec:shangguan} & Lower bound on $\lopt$. Exact $\lopt$ for a special case & Construction~I of~\cite{SZG_IT_18} & \\
\hline
\hline
Sec~\ref{sub:sec:MN-PDA} & Performance analysis for $\ei=O(F)$;  $\log_2 \ei =  O(\log_2 F), o(\log_2 F)$ & Maddah-Ali \& Niesen~\cite{Ali_Niesen_2014} & Lemmas~\ref{lem:MN:sparse_update},~\ref{lem:MN:main:1} \\
\hline
\hline
Sec~\ref{sub:sec:uv_PDA} & Performance analysis for $\ei=O(F)$;  $\log_2 \ei = O(\log_2 F), o(\log_2 F)$ & A scheme from~\cite{YCTC_2017,TaR_IT_18} & Lemmas~\ref{lem:uv_PDA_sparse:1},~\ref{lem:uv_PDA_sparse:2} \\
\Xhline{5\arrayrulewidth}
\end{tabular}
\caption{Summary of the contents of this paper.}
\label{table:summary}
\end{table}

In this paper, we provide a formal framework for the design of blind schemes for updating cache contents. 
Throughout this paper we study only linear coding schemes, and all the converse bounds and results related to optimality of communication costs are with respect to the family of linear coding strategies. 
The main contributions of this paper (summarized in Table~\ref{table:summary}) are as follows.

{\textit{1. Design criterion and a construction based on MDS codes (Section~{\ref{cu_bnsi}}):}} 
We show that the cache update problem is equivalent to a BNSI problem. 
Based on this equivalence, we obtain a design criterion for constructing a linear code for the cache update problem (Theorem~\ref{thrm1}). 
Relying on a construction from~\cite{SuN_2019}, we also identify a simple coding scheme with communication cost of $(F - (Z -2\ei)^+)\,B$ bits that uses the parity-check matrix of MDS codes for encoding, where $B$ is the number of bits in each subfile and $x^+ = \max\{x,0\}$.

{\textit{2. A new code construction (Section~\ref{sec:coding_scheme}):}}
We propose a new coding scheme for the cache update problem with cost $\left(2\ei(K-r) + 1\right)B$ bits. 
The encoder is designed by associating a subspace with each user and by choosing each column of the encoder matrix from the intersection of a carefully chosen collection of these subspaces. 
The proposed construction is random and is over the finite field of size ${2^B}$. 
Using the Schwartz-Zippel lemma, we show that this construction succeeds with high probability for all sufficiently large $B$.

{\textit{3. Converse bounds and performance analysis (Section~\ref{sec:lower_bounds}):}} 
%
We present a lower bound on the optimal communication cost $\lopt$ of linear update schemes (Theorem~\ref{thm:generic_lower_bound}).
Using this bound, we explicitly identify $\lopt$ for a special case of the class of PDAs designed by Shangguan et al.~{\cite{SZG_IT_18}}.
We then specialize our lower bound to two well known families of PDAs -- the Maddah-Ali \& Niesen's scheme~\cite{Ali_Niesen_2014} and a PDA designed independently by Yan et al.~\cite{YCTC_2017} and Tang \& Ramamoorthy~\cite{TaR_IT_18} -- and conduct performance analysis when the number of users $K$ is made large while the caching ratio $Z/F$ is held constant.
For these two families of PDAs, the communication cost of our new scheme from Section~\ref{sec:coding_scheme} is nearly-optimal and equals $\lopt(1 + o(1))$ when $\log_2 \ei = o(\log_2 F)$, i.e., if the updates are sufficiently sparse. When the updates are dense, i.e., when $\ei = O(F)$, the MDS codes based construction is order optimal.

We introduce the system model in Section~{\ref{sec:system_model}} and draw conclusions in Section~{\ref{conclusions}}.

\subsection{Related Work}

We are not aware of any prior coding strategies in the literature that exploit the contents of a deleted file to push blind updates into caching nodes.
However, a problem of similar flavour has been studied in~\cite{PrM_IT_18,NSR_GC_14} for coded distributed storage systems where storage nodes maintain coded versions (i.e., specific linear combinations) of a data file to provide protection against node failures. Communication schemes were designed to update the contents of a `stale' node in the system by either a central server~\cite{PrM_IT_18} or the other nodes in distributed storage~\cite{NSR_GC_14} that have access to the updated content.
Our paper differs fundamentally from the problem setting of~\cite{PrM_IT_18,NSR_GC_14} since we are interested in uncoded placement of subfiles in the caching nodes, i.e., the caching nodes store fragments of the original file instead of linear combinations.

\noindent
\emph{Notation:} Matrices and column vectors will be denoted using bold upper case and bold lower case letters, respectively. For any matrix $\mathbf{A}$, $\colspan(\mathbf{A})$ denotes the column span of $\mathbf{A}$. For any vector $\xb=(x_f~|~f \in \fcal)$ whose entries are indexed by the set $\fcal$ and for any $\acal \subseteq \fcal$, $\xb_{\acal} = (x_f~|~f \in \acal)$ is the subvector of $\xb$ consisting of entries indexed by $\acal$.
For any subspace $U$, $\dim(U)$ denotes the dimension of the subspace.
If $\mathcal{A}$ is a set and $a$ an integer, then $\binom{\mathcal{A}}{a}$ is the collection of all $a$-sized subsets of $\mathcal{A}$. For any positive integer $n$, $[n]$ denotes the set $\{1,\dots,n\}$. The empty set is denoted by $\phi$.
For any real number $x$, $x^+ = \max\{x,0\}$. 

\section{System Model} \label{sec:system_model}

\begin{figure}[t]
\begin{subfigure}{0.5\textwidth}
\centering
\includegraphics[width=3.4in]{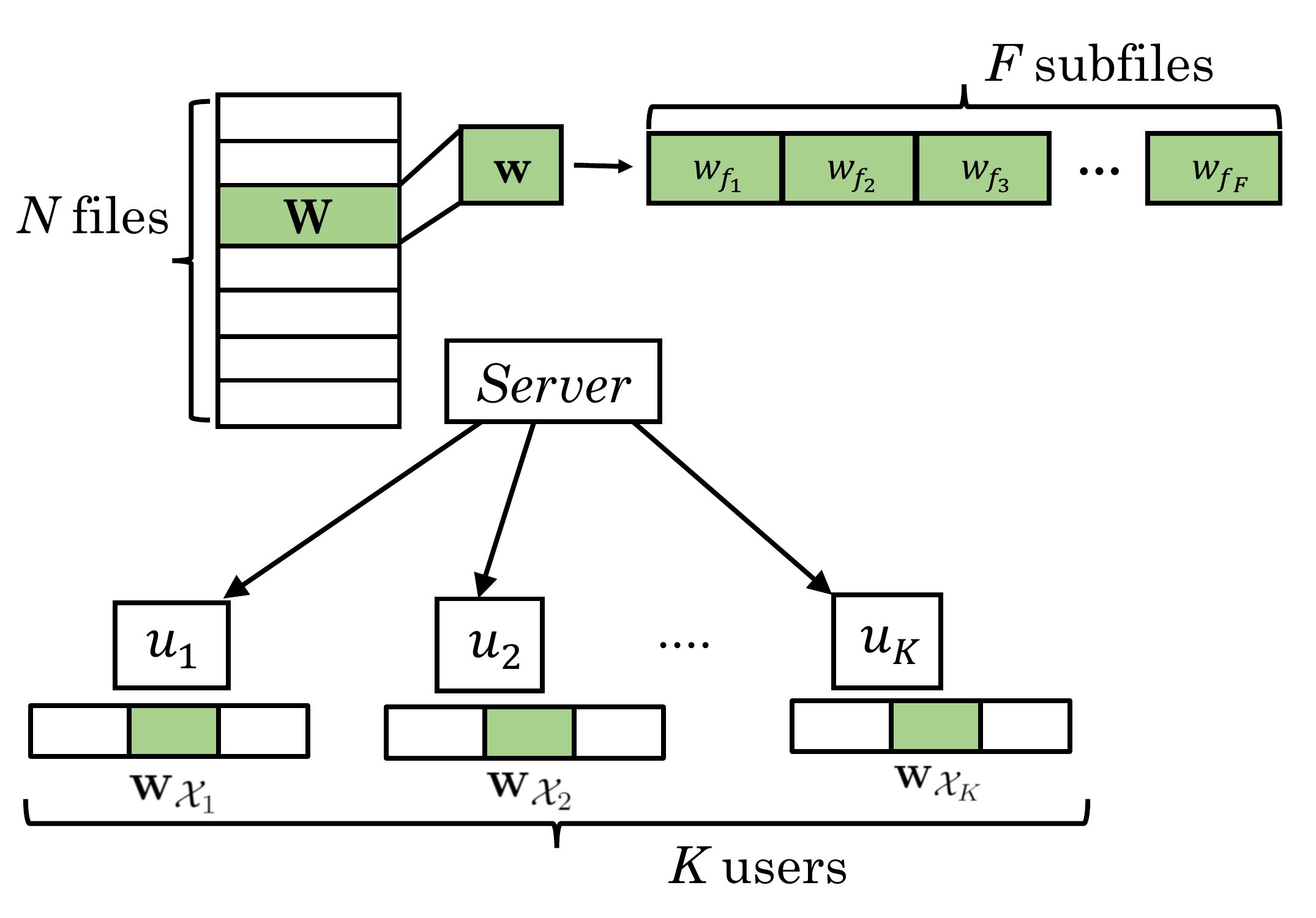} 
\caption{System Model.}
\label{fig:subim1}
\end{subfigure}
\begin{subfigure}{0.5\textwidth}
\centering
\includegraphics[width=3in]{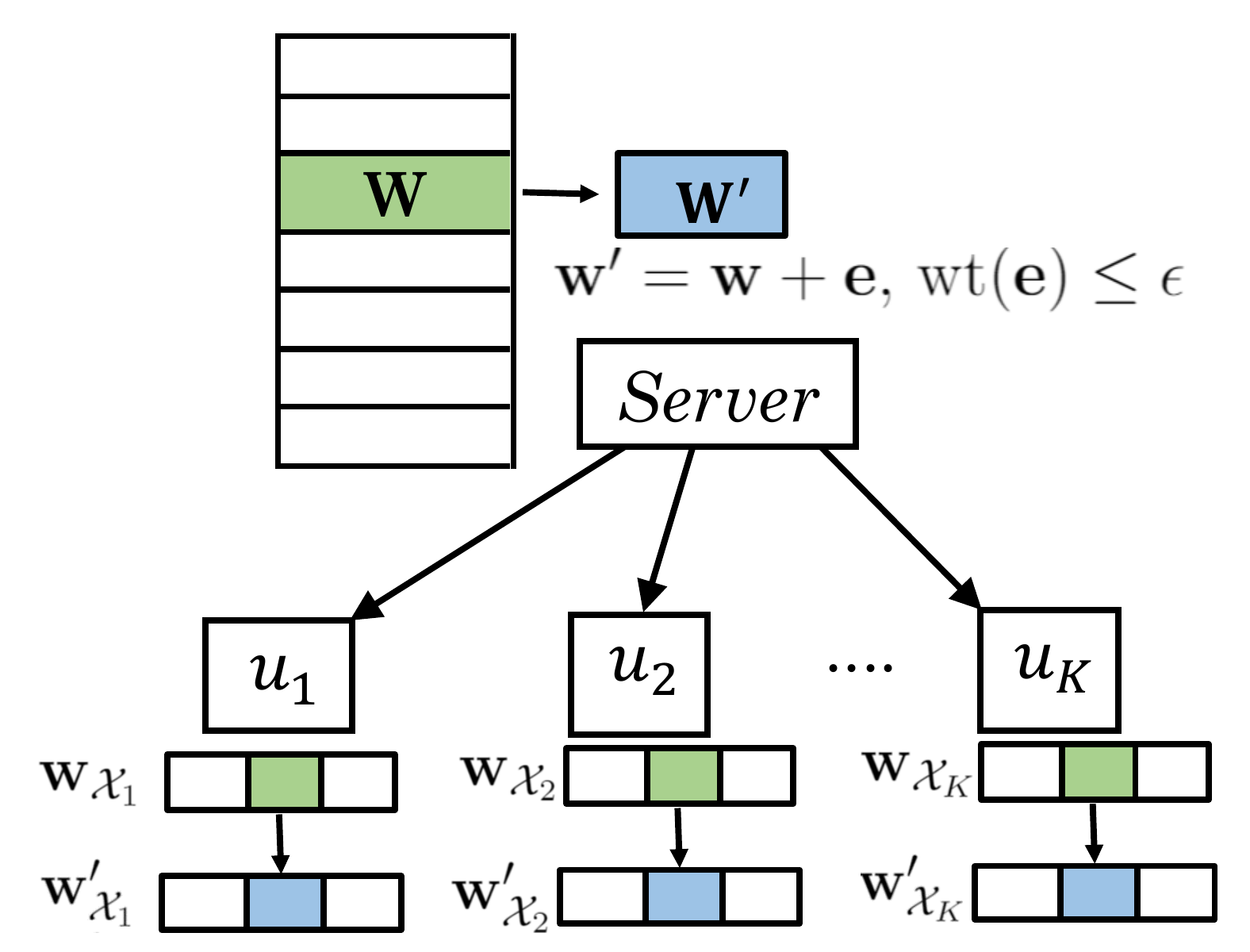}
\caption{Update of a file $W$ to $W'$.}
\label{fig:subim2}
\end{subfigure}
\caption{System model of the cache update problem. Here, the index set of subfiles is $\fcal=\{f_1,\dots,f_F\}$ and the index set of users is $\kcal=\{1,\dots,K\}$.}
\label{fig:image0}
\end{figure}
We consider a caching system with $K$ clients and a single server as shown in Fig~\ref{fig:subim1}. The server stores a library of $N$ files. We view each file as a collection of subfiles, and the subfiles of all the $N$ files are cached in a predetermined way across the $K$ caching nodes.
We assume that one specific file in the library, say $W$, is being replaced with a new file as shown in Fig~\ref{fig:subim2}. The other $N-1$ files existing in the system are unaffected by this change. 
Hence, throughout this paper, we focus only on $W$ and its subfiles and ignore the existence of the other files in the coded caching system. 

The file $W$ is a collection of $F$ equal sized subfiles $W_f,~f \in \fcal$, where $\fcal$ denotes the index set of the subfiles and $|\fcal|=F$. 
The $K$ caching nodes are denoted as $u_k$, $k \in \kcal$, where $\kcal$ is the index set of the nodes and $|\kcal|=K$.
Throughout this paper we will assume that the subfiles are cached across the nodes in an uncoded fashion based on a placement delivery array~\cite{YCTC_2017}.
We also assume that the server is connected to the $K$ nodes through a noiseless broadcast link. 

\subsection{A Brief Review of Placement Delivery Arrays}

A $(K,F,Z,S)$ Placement Delivery Array (PDA) is an $F \times K$ array $P$ whose entries consist of a symbol $\ast$ and $S$ non-negative integers $0,1,\dots,S-1$. The rows and columns of this array are indexed by the sets $\fcal$ and $\kcal$, respectively, i.e., $P = [p_{f,k}]$, $f \in \fcal$ and $k \in \kcal$. The array satisfies the following properties
\begin{enumerate}
\item the symbol $\ast$ appears $Z$ times in each column;
\item each non-negative integer $s \in \{0,1,\dots,S-1\}$ appears at least once in the array; and
\item for any two distinct entries $(f_1,k_1)$ and $(f_2,k_2)$ such that $p_{f_1,k_1}=p_{f_2,k_2}=s$, we have $f_1 \neq f_2$, $k_1 \neq k_2$ and $p_{f_1,k_2}=p_{f_2,k_1}=\ast$.
\end{enumerate} 
We assume that the file $W$ is cached across the $K$ nodes using the strategy based on the PDA $P$, which is as follows: a node $u_k$ caches a subfile $W_f$ if and only if $p_{f,k}=\ast$.

\begin{remark}
The locations of the $\ast$ symbol in a PDA dictate the placement of subfiles in caches, while the integers $\{0,1,\dots,S-1\}$ determine the coded transmissions during the delivery phase of the coded caching system when the clients reveal their file demands to the server. 
In this paper, we are only interested in the specific placement of subfiles in caches since our primary objective is to update the contents of the caches. Hence, we will not be interested in the integer entries of the PDA.
\end{remark}

In this work we consider the family of PDAs in which $\ast$ appears a constant number of times in each row. This family includes several well known PDAs proposed in the literature for coded caching~\cite{Ali_Niesen_2014,YCTC_2017,SZG_IT_18,TaR_IT_18,YTCC_COMML_18} and coded MapReduce~\cite{LMYA_IT_17,KoR_GLOBECOM_18,RaK_ISIT19}.
The number of times $\ast$ appears in each row of the PDA will be denoted by $r$. 
This implies that each node caches $Z$ out of $F$ subfiles and each subfile is replicated at $r$ nodes.
If we count the number of times the symbol $\ast$ appears in each column of the $F \times K$ array $P$ then the total count is $KZ$. Similarly counting the number of times the symbol $\ast$ appears in each row of the array $P$ the total count is $rF$. Now equating these two numbers we get $r=\frac{KZ}{F}$.

The index set of the subfiles cached at node $u_k$ will be denoted as $\xcal_k=\{f \in \fcal~|~p_{f,k}=\ast\}$. Note that $|\xcal_k|=Z$ for all $k \in \kcal$. 
Observe that the tuple $\xcal=(\xcal_k, k \in \kcal)$ completely describes the cache placement strategy.

\begin{example} \label{exmp1}
Suppose we are given the following $(K=4,F=6,Z=3,S=4)$ PDA
\begin{equation*}
{\small
P=
\begin{blockarray}{ccccc}
 & 1 & 2 & 3 & 4\\
\begin{block}{c(cccc)}
\{1,2\} & \ast & \ast & 0 & 1\\
\{1,3\} & \ast & 0 & \ast & 2 \\
\{1,4\} & \ast & 1 & 2 & \ast \\
\{2,3\} & 0 & \ast & \ast & 3\\
\{2,4\} & 1 & \ast & 3 & \ast\\
\{3,4\} & 2 & 3 & \ast & \ast\\
\end{block}
\end{blockarray}~.
}
\end{equation*}
For this PDA, the node index set is $\kcal=[4]=\{1,2,3,4\}$ and the subfile index set is $\fcal = \binom{[4]}{2}$
which is the collection of all $2$-sized subsets of $\kcal$. 
The caching scheme consists of $4$ nodes $u_1,u_2,u_3,u_4$ and $6$ subfiles $W_{\{1,2\}}, W_{\{1,3\}}, W_{\{1,4\}}, W_{\{2,3\}}, W_{\{2,4\}}, W_{\{3,4\}}$. The index set of the subfiles cached at the four nodes are $\xcal_1=\{\{1,2\},\{1,3\},\{1,4\}\}$, $\xcal_2=\{\{1,2\},\{2,3\},\{2,4\}\}$, $\xcal_3=\{\{1,3\},\{2,3\},\{3,4\}\}$ and $\xcal_4=\{\{1,4\}, \{2,4\},\{3,4\}\}$, respectively. 
Note that $|\xcal_k|=Z,~\forall k \in \{1,2,3,4\}$, and each subfile is cached at $r=2$ nodes.
The update problem in Example~\ref{ex:toy_example} is based on this PDA. 
\end{example}

\subsection{The Blind Cache Update Problem}

Let $\wb=(w_f,f \in \fcal) \in \Fb_q^F$ denote the content of the file to be updated where each subfile content, denoted by $w_f$, is an element over a finite field $\Fb_q$. 
The cached content of node $u_k$ is denoted as $\wb_{\xcal_k}=(w_f~|~f \in \xcal_k)$. We observe that $\wb_{\xcal_k} \in \Fb_q^Z$ for all $k \in \kcal$. 

Suppose the file content $\wb$ is updated to $\wb+\eb$, where $\eb \in \Fb_q^F$ represents an update to the file content, i.e., $\wtm(\eb) \leq \ei$ where $\wtm()$ denotes the Hamming weight of a vector and $\ei$ is a known constant. In other words the original file undergoes an update where at the most $\ei$ subfiles have been updated. The update is modeled as a substitution of the original collection of the subfiles. 
We assume that the server is unaware of the old version of the file $\wb$ and the identities of updated subfiles. The server knows the updated file content $\wb+\eb$ and knows that the number of the updated subfiles is at the most $\ei$. {The parameter $\ei$ characterizes the update process. For example, if the ratio $\frac{\ei}{F}$ is small and tends to $0$ as $F \to \infty$ then we view such an update as a sparse update. On the other hand if $\ei$ is proportional to $F$ then the update is dense.}

The server wants to communicate the update to the users, so that each user $u_k$ can update its contents from $\wb_{\xcal_k}$ to $(\wb+\eb)_{\xcal_k}$.
Considering $\ei$ as a design parameter, we are interested in the problem of designing a coding scheme to update the cache of each user described by the cache placement strategy $\xcal=(\xcal_k, k \in \kcal)$ with an update of at the most $\ei$ subfiles. We will call this problem the $(\xcal,\ei)$ \textit{cache update problem} or the $(\xcal,\ei)$ \textit{update problem}.
\begin{definition}
A valid encoding function of codelength $l$ for the $(\xcal,\ei)$ update problem over the field $\Fb_q$ is a function
$\mathfrak{E}: \Fb_q^{F} \to \Fb_q^l$
such that for each node $u_k,~k \in \kcal$ there exists a decoding function $\mathfrak{D}_k: \Fb_q^l \times \Fb_q^Z \to \Fb_q^Z$ satisfying the following property: $\mathfrak{D}_k(\mathfrak{E}(\wb+\eb),\wb_{\xcal_k})=(\wb+\eb)_{\xcal_k}$ for every $\wb \in \Fb_q^{F}$ and $\eb \in \Fb_q^{F}$ with $\wtm(\eb) \leq \ei$.
\end{definition}
The communication cost of the coding scheme, in number of bits, is $l \times \log_2 q$ bits.
The objective of the code construction is to design the coding scheme $(\mathfrak{E},\mathfrak{D}_k, k \in \kcal)$ such that the codelength $l$ is minimized. 

A coding scheme $(\mathfrak{E},\mathfrak{D}_1,\dots,\mathfrak{D}_K)$ is said to be linear if the encoding function is an $\Fb_q$-linear transformation. For a linear coding scheme, the transmitted codeword $\cb=\Hb(\wb+\eb)$, where $\Hb \in \Fb_q^{l \times F}$ is the encoder matrix. 
%
The optimum communication cost among all valid linear coding schemes (considering all possible choices of the finite field $\Fb_q$) for the $(\xcal,\ei)$ update problem will be denoted as $\lopt(\xcal,\ei)$ or simply $\lopt$. 

We assume that the code designer has the flexibility to choose the operating finite field when constructing the update scheme. 
Our code construction in Section~\ref{sec:coding_scheme} is applicable for all sufficiently large finite fields, while our lower bounds in Section~\ref{sec:lower_bounds} are independent of the choice of the finite field.

The trivial coding scheme that transmits the updated file content $\wb+\eb$ as such, i.e., $\cb =
\Ib_{F \times F}(\wb + \eb)$ is a valid coding scheme with codelength $F$ since each node can directly update its cache contents using $\cb$. We refer to this trivial coding scheme as the \textit{naive scheme} where $\Hb=\Ib$. Thus, we have the following
trivial upper bound on the optimum linear codelength
\begin{equation} \label{eq:trivial_upper_bound}
\lopt \leq F.
\end{equation}
 
\section{Equivalence between Cache Update\\ and Broadcasting with Noisy Side Information} \label{cu_bnsi}

In this section we show that every cache update problem is equivalent to a \emph{broadcasting with noisy side information (BNSI)} problem~\cite{SuN_2019}. We start this section with a brief introduction of BNSI problem. We show that by performing a simple change of variables we obtain an instance of BNSI problem starting from an instance of cache update problem. At the end of this section we briefly describe some preliminary results, derived for BNSI problem, in our notation that will help to design linear codes for the cache update problem.

\subsection{Broadcasting with Noisy Side Information (BNSI) Problem} \label{bnsi}

The BNSI problem deals with broadcasting a vector of $F$ information symbols denoted as \mbox{$\xb=(x_f~|~f \in \fcal) \in \Fb_q^F$} from a server to $K$ users denoted as $u_k, k \in \kcal$ over a noiseless broadcast channel. We are given subsets $\xcal_k \subseteq \fcal$, $k \in \kcal$, such that the user $u_k$ demands the subvector $\xb_{\xcal_k} = (x_f, f \in \xcal_k)$. 
We will assume that $|\xcal_k|=Z$ for each $k$. 
Further, each user $u_k$ already knows an erroneous version of its demand as side information, i.e., knows the value of $\xb_{\xcal_k} + \bxi_k$ where $\bxi_k \in \Fb_q^{Z}$ is an unknown noise vector satisfying $\wtm(\bxi_k) \leq \ei$.
The users and the server do not know the exact values of the side information error vector, but are aware that the side information at any user and its demand differ in at the most $\ei$ coordinates.
The aim of the problem is to design an \emph{encoding function} $\mathfrak{E}: \Fb_q^F \rightarrow \Fb_q^l$
such that for each user $u_k$ there exists a \emph{decoding function}
\mbox{$\mathfrak{D}_k:\Fb_q^l \times \Fb_q^{Z} \rightarrow \Fb_q^{Z}$}
satisfying the property $\mathfrak{D}_k(\mathfrak{E}(\xb),\xb_{\xcal_k}+\bxi_k)=\xb_{\xcal_k}$ for every \mbox{$\xb \in \Fb_q^F$} and $\bxi_k \in \Fb_q^{Z}$ with $\wtm(\bxi_k) \leq \ei$. 
Defining \mbox{$\xcal=(\xcal_k, k \in \kcal)$} to be the $K$-tuple that represents the demands of all the $K$ users, we refer to this communication problem as the \emph{$(\xcal,\ei)$ BNSI (Broadcasting with Noisy Side Information) problem}. 

\subsection{Equivalence Between Cache Update and BNSI Problems} \label{cacheupdate_bnsi}

In this subsection we make the observation that any coding scheme for the $(\xcal,\ei)$ cache update problem is a solution for the $(\xcal,\ei)$ BNSI problem, and vice-versa. 

We observe that in both problems each user demands a particular subvector of an $F$-length vector from the transmitter. In the BNSI problem the $F$-length vector available at the source is $\xb$ and $u_k$ demands $\xb_{\xcal_k}$, while in the cache update problem the transmitter has $\wb + \eb$ and $u_k$ demands $(\wb + \eb)_{\xcal_k} = \wb_{\xcal_k} + \eb_{\xcal_k}$.
In the BNSI problem as well as the cache update problem, the side information available at each user is a noisy version of its own demand. In the former problem, $u_k$ demands $\xb_{\xcal_k}$ and knows $\xb_{\xcal_k} + \bxi_k$ as side information, $\bxi_k$ being the noise at $u_k$. In the latter problem, $u_k$ demands $\wb_{\xcal_k} + \eb_{\xcal_k}$ and knows $\wb_{\xcal_k}$. The difference between the side information and demand in the cache update problem, viz. $\wb_{\xcal_k} - \left(\wb_{\xcal_k} + \eb_{\xcal_k} \right) = -\eb_{\xcal_k}$ is the effective noise at $u_k$. 

The noise affecting the $K$ users in the cache update problem, $-\eb_{\xcal_k}, k \in \kcal$, are all subvectors of the negative of the update vector $\eb \in \Fb_q^F$.
In contrast, the noise vectors affecting the users in the BNSI problem $\bxi_1,\dots,\bxi_K$ are arbitrary and could be independent of each other. 
This is the key difference between the two problems. However, in spite of this difference, we next observe that the coding solutions to both the problems are identical. The reason why the dependence or correlation of the noise vectors in the cache update problem does not provide any additional coding leverage is because the communication channel is a broadcast link and the users do not collude or cooperate during decoding.

We first show that any coding scheme for the $(\xcal,\ei)$ cache update problem is a valid coding scheme for the $(\xcal,\ei)$ BNSI problem. 
Let $\mathfrak{E}$, $\mathfrak{D}_k, k \in \kcal$ be valid encoding and decoding functions for the cache update problem, that is 
\begin{equation} \label{eq:equivalence:1}
\mathfrak{D}_k(\mathfrak{E}(\wb + \eb),\wb_{\xcal_k}) = \wb_{\xcal_k} + \eb_{\xcal_k} = (\wb + \eb)_{\xcal_k}
\end{equation} 
for all choices of $\wb$ and $\eb$ with $\wtm(\eb) \leq \ei$ and for every user $u_k$. 
We know that the noise vectors $\bxi_k \in \Fb_q^Z$, $k \in \kcal$, in the BNSI problem have weight at the most $\ei$. For each $k \in \kcal$, define the vector $\eb^{(k)} \in \Fb_q^F$ as follows, $\eb^{(k)}_{\xcal_k} = -\bxi_k$ and $\eb^{(k)}_f=0$ for all $f \notin \xcal_k$. Observe that $\wtm(\eb^{(k)}) \leq \ei$ and $\xb_{\xcal_k} + \bxi_k = \xb_{\xcal_k} - \eb^{(k)}_{\xcal_k} = (\xb - \eb^{(k)})_{\xcal_k}$. When the coding scheme $(\mathfrak{E},\mathfrak{D}_k, k \in \kcal)$ is applied to the BNSI problem, for each $k \in \kcal$, we have
\begin{align*}
\mathfrak{D}_k(\mathfrak{E}(\xb),\xb_{\xcal_k} + \bxi_k) 
= \mathfrak{D}_k\left(\,\mathfrak{E}\left((\xb - \eb^{(k)}) + \eb^{(k)}\right),(\xb - \eb^{(k)})_{\xcal_k}\,\right)
= ((\xb - \eb^{(k)}) + \eb^{(k)})_{\xcal_k} = \xb_{\xcal_k},
\end{align*} 
where the second equality follows from~\eqref{eq:equivalence:1} with $(\xb - \eb^{(k)})$ and $\eb^{(k)}$ playing the roles of $\wb$ and $\eb$, respectively.
Hence, $(\mathfrak{E},\mathfrak{D}_k, k \in \kcal)$ is a valid coding scheme for the $(\xcal,\ei)$ BNSI problem.

Conversely, now assume that $\mathfrak{E}$ and $\mathfrak{D}_k,k \in \kcal$, are valid encoding and decoding functions for the $(\xcal,\ei)$ BNSI problem. That is, for every $u_k$ and any choice of $\xb$, $\bxi_k$, $k \in \kcal$ with $\wtm(\bxi_k) \leq \ei$,
\begin{equation} \label{eq:equivalence:2}
\mathfrak{D}_k(\mathfrak{E}(\xb),\xb_{\xcal_k} + \bxi_k) = \xb_{\xcal_k}.
\end{equation} 
For given values of $\wb$ and $\eb$ in the cache update problem, define $\xb = \wb + \eb$ and $\bxi_k = -\eb_{\xcal_k}$ for all $k \in \kcal$.
We know that $\wtm(\eb) \leq \ei$, and hence, $\wtm(\bxi_k) = \wtm(-\eb_{\xcal_k}) \leq \wtm(-\eb) \leq \ei$. 
Also, $\wb_{\xcal_k} = \xb_{\xcal_k} - \eb_{\xcal_k} = \xb_{\xcal_k} + \bxi_k$.
From~\eqref{eq:equivalence:2}, and using the fact $\wtm(\bxi_k) \leq \ei$, we have
\begin{align*}
\mathfrak{D}_k(\mathfrak{E}(\wb + \eb),\wb_{\xcal_k} ) 
= \mathfrak{D}_k(\mathfrak{E}(\xb), \xb_{\xcal_k} + \bxi_k ) = \xb_{\xcal_k} 
= \wb_{\xcal_k} + \eb_{\xcal_k}. 
\end{align*} 
Hence, $(\mathfrak{E},\mathfrak{D}_k, k \in \kcal)$ is a valid coding scheme for the cache update problem.

The equivalence proved above holds for all codes, including linear and non-linear codes. 
Suppose linear codes are used, i.e., the codeword is generated at the transmitter by multiplying the $F$-length information vector with an $l \times F$ matrix $\Hb$. We say that the encoding matrix $\Hb$ is \emph{valid} for the given problem (either the BNSI or the cache update problem) if every receiver can decode its demand using the codeword generated by $\Hb$ and its own side information.
Applying the equivalence proved in this subsection to linear codes we obtain

\begin{theorem} \label{thm:cache_update_BNSI}
A matrix $\Hb$ is a valid encoder for the $(\xcal,\ei)$ cache update problem if and only if it is a valid encoder for the $(\xcal,\ei)$ BNSI problem.
\end{theorem}

\subsection{Preliminaries} \label{prep}

We now recall some relevant results from \cite{SuN_2019}, including a necessary and sufficient condition for a matrix $\Hb$ to be a valid encoder for the BNSI problem. 
Since the BNSI problem is equivalent to the cache update problem, we directly state these results as applied to the cache update problem. We then recall a construction of the encoder matrix from~\cite{SuN_2019} based on Maximum Distance Separable (MDS) codes.

The linear code design criterion of~\cite{SuN_2019} is in terms of the span of the columns of specific submatrices of $\Hb$. For each node $u_k,~k \in \kcal$, let $\ycal_k=\fcal \setminus \xcal_k$ denote the index set of subfiles that are not cached by the node $u_k$. Since $|\xcal_k|=Z$, we have $|\ycal_k|=F-Z$.

We index the $F$ columns of $\Hb$ by the elements of $\fcal$. For any $\acal \subseteq \fcal$, let $\Hb_{\acal} \in \Fb_q^{l \times |\acal|}$ be the submatrix of $\Hb$ consisting of the columns of $\Hb$ with indices belonging to $\acal$. Also, let $\colspan(\Hb_{\acal})$ denote the subspace of $\Fb_q^l$ spanned by the columns of $\Hb_{\acal}$.

\begin{theorem}~\cite[Corollary~1]{SuN_2019} \label{thrm1}
A matrix $\Hb$ is a valid encoder matrix for the $(\xcal,\ei)$ cache update problem if and only if for every $k \in \kcal$, any non-zero linear combination of any $2\ei$ or fewer columns of $\Hb_{\xcal_k}$ does not belong to $\colspan(\Hb_{\ycal_k})$, i.e.,
\begin{equation*}
 \Hb_{\xcal_k}{\bf x} + \Hb_{\ycal_k}{\bf y} \neq {\bf 0} \text{ for all } {\bf x} \in \Fb_q^Z \setminus \{{\bf 0}\} \text{ with } \wtm({\bf x}) \leq 2\ei \text{ and } {\bf y} \in \Fb_q^{F-Z}.
\end{equation*} 
\end{theorem}

Since the column span $\colspan(\Hb_{\ycal_k})$ includes $\pmb{0}$, Theorem~\ref{thrm1} implies that any $2\ei$ or fewer columns of $\Hb_{\xcal_k}$ must be linearly independent.

\begin{corollary} \cite[Corollary~2]{SuN_2019} \label{corr1}
If $\Hb$ is a valid encoder matrix for the $(\xcal,\ei)$ cache update problem then any $2\ei$ or fewer columns of $\Hb_{\xcal_k}$ are linearly independent for every $k \in \kcal$.
\end{corollary}

\begin{example}
Consider the $(\xcal,\ei=1)$ cache update problem, with $K=4$ users and $F=6$ subfiles and the cache placement as given in Example~\ref{exmp1} over any finite field $\Fb_q$. 
Here the subfiles are indexed by all $2$-subsets of $[4]=\{1,\dots,4\}$, i.e., $\fcal=\binom{[4]}{2}$.
Now consider 
\begin{equation*}
\Hb=
\begin{blockarray}{cccccc}
 \{1,2\} & \{1,3\} & \{1,4\} & \{2,3\} & \{2,4\} & \{3,4\}\\
\begin{block}{(cccccc)}
1 & 0 & 0 & 0 & 0 & 1\\
0 & 1 & 0 & 0 & 0 & 1\\
0 & 0 & 1 & 0 & 0 & 1\\
0 & 0 & 0 & 1 & 0 & 1\\
0 & 0 & 0 & 0 & 1 & 1\\
\end{block}
\end{blockarray}~.
\end{equation*}
Note that any five columns of $\Hb$ are linearly independent. 
For user $u_1$, $\xcal_1=\left\{ \{1,2\}, \{1,3\}, \{1,4\} \right\}$ and $\ycal_1 = \fcal \setminus \xcal_1 = \left\{ \{2,3\}, \{2,4\}, \{3,4\} \right\}$.
Observe that any two columns of $\Hb_{\xcal_1}$ are linearly independent, and any non-zero linear combination of any two columns of $\Hb_{\xcal_1}$ does not belong to $\colspan(\Hb_{\ycal_1})$.
Similar observations hold for users $u_2,u_3$ and $u_4$ as well.
Hence, $\Hb$ is a valid encoder matrix and achieves codelength $l=5$ that saves $1$ channel use with respect to the naive update scheme.
This is the coding scheme used in Example~\ref{ex:toy_example}.
\end{example}

The following result from~\cite{SuN_2019} will be used in Section~\ref{sub:sec:converse:general} to identify the scenarios where coding can reduce the communication cost with respect to the naive transmission scheme.

\begin{theorem} \label{thm:lower_bound_XS}
\cite[Theorem~3]{SuN_2019} Let $S = \{ k \in \kcal~|~|\xcal_k| \leq 2\ei\}$ be the index set of all nodes that cache $2\ei$ or fewer subfiles. 
Let $\xcal_S = \cup_{k \in S} \xcal_k$ be the index set of all subfiles cached among the nodes in $S$. Then 
$\lopt \geq |\xcal_S| + \min\{2\ei,F-|\xcal_S|\}$.
\end{theorem}

In the cache update setting $|\xcal_k|=Z$ for all $k$. Hence, in Theorem~\ref{thm:lower_bound_XS}, we either have $S=\phi$ or $S=\kcal$, and correspondingly, $\xcal_S=\phi$ or $\xcal_S=\fcal$.

We now recall a coding scheme from~\cite{SuN_2019} that relies on MDS codes. In this scheme, we choose $\Hb$ to be the parity-check matrix of an MDS code of length $F$ and dimension $(Z-2\ei)^+$. The number of rows $l$ of $\Hb$ is $F - (Z-2\ei)^+$, and from the properties of MDS codes we know that any $l$ columns of $\Hb$ are linearly independent. 
Note that the number of columns of $\Hb_{\xcal_k}$ is $|\xcal_k|=Z$.
To check if $\Hb$ satisfies the criteria of Theorem~\ref{thrm1}, consider any $k \in \kcal$ and the union of any $\min\{2\ei,Z\}$ columns of $\Hb_{\xcal_k}$ and all the columns of $\Hb_{\ycal_k}$. 
The total number of columns in this union is $\min\{2\ei,Z\} + |\ycal_k| = \min\{2\ei,Z\} + F - Z = F - (Z-2\ei)^+ =l$. Hence, these columns are linearly independent and satisfy the criteria of Theorem~\ref{thrm1}.
Such an MDS code exists over $\Fb_q$ if $q \geq F$.
Hence, we have the following upper bound on $\lopt$.

\begin{lemma} \label{ub1}
The optimal codelength $\lopt$ of a $(\xcal,\ei)$ cache update problem satisfies $\lopt \leq F-(Z-2\ei)^+$.
\end{lemma}

This code design provides savings in communication cost with respect to naive transmission, i.e., has codelength $l < F$ if and only if $Z \geq 2\ei + 1$.

\section{A Scheme for Blind Updates in Coded Caching} \label{sec:coding_scheme}

In this section we provide a construction of a linear coding scheme for the $(\xcal,\ei)$ update problem arising from a PDA. 
We will assume that the PDA satisfies the following condition
\begin{equation} \label{eq:mild_condition_pda}
 \{k~|~p_{f_1,k} = \ast\} ~\neq~ \{k~|~p_{f_2,k} = \ast\} \text{ for any } f_1 \neq f_2,
\end{equation} 
that is, for any two distinct subfiles $W_{f_1}$ and $W_{f_2}$ the set of nodes storing $W_{f_1}$ and the set of nodes storing $W_{f_2}$ are distinct.
Several popular families of PDAs satisfy this condition, such as~\cite{SZG_IT_18,YTCC_COMML_18,TaR_IT_18,Ali_Niesen_2014}.
The communication cost of our coding scheme is $l = 2\ei(K-r) + 1$, where $r$ is the number of times $\ast$ appears in each row of the PDA.
The construction is random and yields a valid encoder matrix with probability $1 - O(q^{-1})$ when designed over the finite field $\Fb_q$.
If we use a sufficiently large finite field, then this probability is non-zero, and hence, this proves the existence of a valid code.
For ease of exposition, and considering the engineering significance, we will consider only finite fields of characteristic $2$.
The main result of this section is 

\begin{theorem} \label{thm:new_coding_scheme}
Over every sufficiently large finite field of characteristic $2$ there exists a valid linear code for the $(\xcal,\ei)$ update problem with codelength $l=2\ei(K-r)+1$ if the PDA satisfies~\eqref{eq:mild_condition_pda}.
\end{theorem} 

Combining this result with Lemma~\ref{ub1}, for any PDA satisfying~\eqref{eq:mild_condition_pda} we have
\begin{equation} \label{eq:joint_upper_bound}
\lopt \leq \min\{2\ei(K-r)+1,F - (Z-2\ei)^+\}.
\end{equation}

We provide an overview of the construction in Section~\ref{sec:sub:construction}, and in Section~\ref{sec:sub:proofs_construction} prove that this construction yields a valid code with high probability for large finite fields.

Throughout this section we will consider the parameters $K,F,Z,S,r$ of the PDA and the update parameter $\ei$ as constants, and treat the field size $q$ as a variable. We will say that an event occurs with \emph{high probability} if its probability is at least $1- O(q^{-1})$. 

\subsection{Construction of the Encoder Matrix} \label{sec:sub:construction} 

Let $l = 2\ei(K-r) + 1$. We pick $K$ carefully designed random subspaces $V_k, k \in \kcal$, one corresponding to each user, independently of each other. The construction of $V_k$ will be described later in this sub-section. The subspaces $V_k$, $k \in \kcal$, will be chosen such that their dimension is $l - 2\ei$ with high probability.
For any collection $\ical \subset \kcal$ of nodes, we define $V_{\ical} = \cap_{k \in \ical} V_k$.

For each subfile index $f \in \fcal$, let 
\begin{equation*}
\ical_f = \{ k \in \kcal~|~p_{f,k} \neq \ast\} = \{ k \in \kcal~|~ f \notin \xcal_k\}
\end{equation*} 
denote the set of nodes which do not cache the subfile $W_f$. Thus $|\ical_f|$ is the number of non-$\ast$ entries in row $f$ of the PDA, and hence, $|\ical_f|=K-r$ for all $f \in \fcal$. 
Note that 
\begin{equation} \label{eq:V_If_definition}
V_{\ical_f} = \bigcap_{k \in \ical_f} V_k = \bigcap_{k: p_{f,k} \neq \ast} V_k.
\end{equation} 
 
With high probability, the subspaces $V_{\ical_f}, f \in \fcal$ will be $1$-dimensional (see Section~\ref{sec:sub:proofs_construction}). 
To construct the encoder matrix, we choose the $f^{\text{th}}$ column of $\Hb$, denoted as $\hb_f$, to be any non-zero vector in the $1$-dimensional subspace $V_{\ical_f}$. That is, the $F$ columns of $\Hb=[\hb_f]_{f \in \fcal}$ are the basis vectors of the subspaces $V_{\ical_f}, f \in \fcal$, respectively.

\begin{example}
\begin{figure}[!t]
\centering
\includegraphics[width=5in]{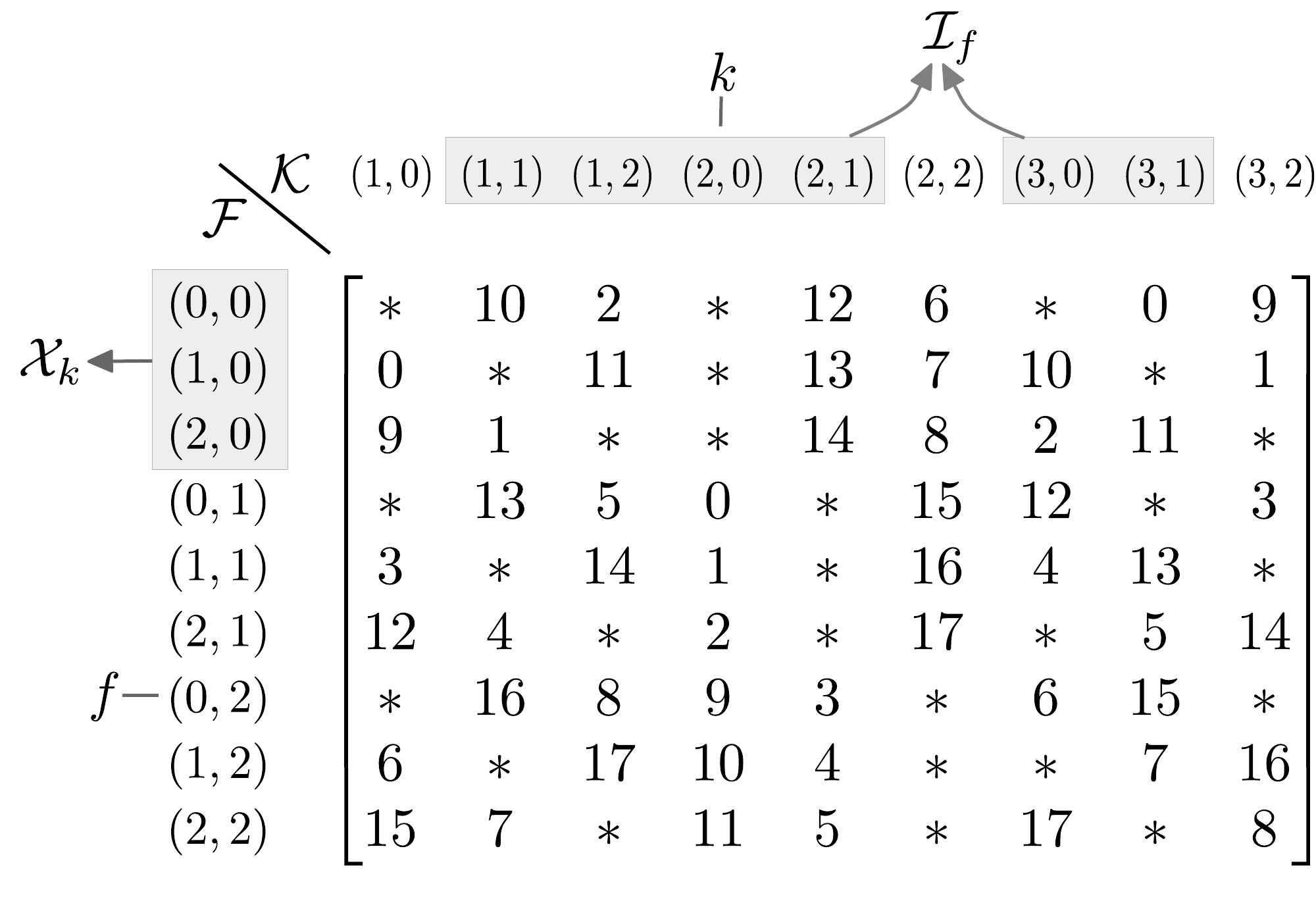}
\caption{The $9 \times 9$ matrix (enclosed within square brackets) is the PDA called $\Ab^{(3,2)}$ from~\cite{YCTC_2017}. The parameters of this PDA are $K=9$, $F=9$, $Z=3$, $S=18$, $r=3$.
The columns and rows of the PDA are indexed by $2$-tuples as shown, and this notation arises from the construction technique of this PDA. 
The set of column indices is $\kcal$ which represents the caching nodes, and the set of row indices is $\fcal$ which denotes the set of subfiles.
The subfiles cached at node $k=(2,0)$ is $\xcal_k=\{(0,0),\,(1,0),\,(2,0)\}$. The nodes that \emph{do not} contain the subfile $f=(0,2)$ is $\ical_f=\{(1,1),(1,2),(2,0),(2,1),(3,0),(3,1)\}$. Our coding scheme assigns a randomly generated subspace $V_k$ to each column $k$ of the PDA. 
Using these $K$ subspaces, we generate $F$ subspaces, one corresponding to each row of the PDA. The subspace assigned to row $f$ is $V_{\ical_f} = \cap_{k \in \ical_f} V_k$.
}
\label{fig:A32_PDA}
\hrule
\end{figure} 
See Fig.~\ref{fig:A32_PDA} for a graphical description of the assignment of subspaces $V_{\ical_f}$ to each subfile $f \in \fcal$ for the PDA from~\cite[Table~III]{YCTC_2017}.
\end{example}

\begin{example}
Consider the PDA in Example~\ref{exmp1} and let $\ei=1$. For this PDA we know that $K=4$, $F=6$ and $r=2$. Recall that $\kcal = \{1,2,3,4\}$ and $\fcal = \binom{[4]}{2}$, i.e., each $f \in \fcal$ is a $2$-subset of $\{1,2,3,4\}$.
By inspection, we see that $\ical_f = \kcal \setminus f$, i.e., $\ical_{\{1,2\}}=\{3,4\}$, $\ical_{\{1,3\}} = \{2,4\}$ etc.

To apply our construction, we use $l=2\ei(K-r) + 1 = 5$. For each $k \in \{1,\dots,4\}$, $V_k$ will be a random subspace of $\Fb_q^5$, with dimension equal to $l-2\ei=3$ with high probability. For each of the six subfiles $f$, we choose the subspace $V_{\ical_{f}}$ as follows $V_{\ical_{\{1,2\}}}=V_3 \cap V_4$, $V_{\ical_{\{1,3\}}}=V_2 \cap V_4$, $V_{\ical_{\{1,4\}}}=V_2 \cap V_3$, $V_{\ical_{\{2,3\}}}=V_1 \cap V_4$, $V_{\ical_{\{2,4\}}}=V_1 \cap V_3$ and $V_{\ical_{\{3,4\}}}=V_1 \cap V_2$.
Finally, for each $f \in \fcal$, we pick $\hb_f$ to be any non-zero vector in the subspace $V_{\ical_f}$. The encoder matrix is 
$\Hb =
\begin{bmatrix}
\hb_{\{1,2\}} & \hb_{\{1,3\}} &\hb_{\{1,4\}} &\hb_{\{2,3\}} &\hb_{\{2,4\}} &\hb_{\{3,4\}} 
\end{bmatrix}$.
\end{example}

We now describe how the subspaces $V_k$, $k \in \kcal$, are chosen. We utilize a set of $2\ei K$ random scalars $a_{k,m}$, $k \in \kcal$, $m \in [2\ei]$, which are independent and uniformly distributed over $\Fb_q$. 
We define $V_k$ through its orthogonal complement $V_k^\perp$ as follows:
$V_k^\perp$ is the column space of the $l \times 2\ei$ Vandermonde matrix generated by $A_k \triangleq \{a_{k,1},a_{k,2},\dots,a_{k,2\ei}\}$,
\begin{equation*}
\vand(A_k) \triangleq \begin{bmatrix}
1 & 1 & \cdots & 1 \\
a_{k,1} & a_{k,2} & \cdots & a_{k,2\ei} \\
\vdots & \vdots &   & \vdots \\
a_{k,2\ei}^{l-1} & a_{k,2}^{l-1} & \cdots & a_{k,2\ei}^{l-1}
\end{bmatrix},
\end{equation*}
that is, $V_k^\perp = \colspan(\vand(A_k))$. 
The scalars $a_{k,1},\dots,a_{k,2\ei}$ will be distinct with probability $1- O(q^{-1})$, and hence, $\dim(V_k^\perp) = 2\ei$ and $\dim(V_k) = l-2\ei$ with high probability.

\subsection{Technical Proofs} \label{sec:sub:proofs_construction}

We first prove a general result on random Vandermonde matrices.

\begin{lemma} \label{lem:random_vandermonde}
Let $U$ be an arbitrary $2\ei$-dimensional subspace of $\Fb_q^l$ and $V_k^\perp=\colspan(\vand(A_k))$, where $A_k = \left\{a_{k,1},\dots,a_{k,2\ei}\right\}$.
If $a_{k,1},\dots,a_{k,2\ei}$ are chosen independently and uniformly at random from $\Fb_q$, then $U \cap V_k = \{\pmb{0}\}$ with probability $1-O(q^{-1})$. 
\end{lemma}
\begin{proof}
We will treat $l$-dimensional vectors ${\bf p}=(p_1,\dots,p_l)^T$ as polynomials $p(x) = p_1 + p_2 x + p_3 x^2 + \cdots + p_lx^{l-1}$ of degree at the most $l-1$. 
A non-zero vector ${\bf p} \in V_k$ if and only if ${\bf p}^T \vand(A_k) = {\bf 0}^T$, i.e., $a_{k,1},\dots,a_{k,2\ei}$ are all roots of $p(x)$. If the degree of $p(x)$ is $d$, then $p(x)$ has at the most $d$ distinct roots in $\Fb_q$. 
Hence, the probability that a randomly chosen scalar $a_{k,m}$ is a root of $p(x)$ is at the most $d/q$. Since $a_{k,1},\dots,a_{k,2\ei}$ are independent random variables, the probability that \emph{all} of them are roots of $p(x)$ is at the most $(d/q)^{2\ei} \leq ((l-1) \, / \, q)^{2\ei}$.

We will prove the lemma by showing that the probability that $U \cap V_k$ contains a non-zero polynomial is $O(q^{-1})$. 
Since $U \cap V_k$ is a subspace, it contains a non-zero polynomial if and only if it contains a \emph{monic} polynomial.
Using the fact that there are exactly $(q^{2\ei}-1)/(q-1)$ monic polynomials in $U$, we have
\begin{align*}
P[U \cap V_k \neq \{\pmb{0}\}] 
&= P \left[ \bigcup_{\substack{p(x) \in U \setminus \{\pmb{0}\} }} \{p(x) \in V_k \}  \right] 
= P \left[ \bigcup_{\substack{p(x) \in U \\ p(x) \text{ is monic} }} \{p(x) \in V_k \}  \right] \\
&\leq \sum_{\substack{p(x) \in U \\ p(x) \text{ is monic} }} \!\!\!\!\!\!\! P\left[ p(x) \in V_k \right] 
\leq \frac{q^{2\ei}-1}{q-1} \, \left(\frac{l-1}{q}\right)^{2\ei} 
= O(q^{-1}).
\end{align*} 
\end{proof}

Towards showing that our construction succeeds we will now show that the subspaces $V_{\ical_f}, f \in \fcal$ are one-dimensional with high probability.

\begin{lemma}
The subspaces $V_{\ical_f}$, $f \in \fcal$, are $1$-dimensional with probability $1 - O(q^{-1})$.
\end{lemma}
\begin{proof}
From the definition~\eqref{eq:V_If_definition} of $V_{\ical_f}$, we have $V_{\ical_f}^\perp = \sum_{k \in \ical_f} V_k^\perp$, which is the subspace obtained by the sum of the column spaces of $\vand(A_k)$, $k \in \ical_f$. That is, $V_{\ical_f}^\perp$ is the column space of $\vand(\cup_{k \in \ical_f} A_k)$ which is the $l \times 2\ei(K-r)$ Vandermonde matrix generated by the scalars $\cup_{k \in \ical_f}A_k = \left\{a_{k,m}~|~k \in \ical_f, m \in [2\ei] \right\}$. This matrix has rank $2\ei(K-r)=l-1$, and hence, $\dim(V_{\ical_f})=1$, as long as all the scalars in $\cup_{k \in \ical_f}A_k$ are distinct. The probability that any two random variables in $\cup_{k \in \ical_f}A_k$ take the same value is $O(q^{-1})$.
Thus the proof is complete.
\end{proof}

We will denote the set of $2\ei(K-r)$ scalars $\cup_{k \in \ical_f} A_k$ by $A_{\ical_f}$. We will use the notation $a_{k,m} \in A_{\ical_f}$ to imply that $a_{k,m}$ is such that $k \in \ical_f$ and $m \in [2\ei]$. 

We will now explicitly identify a basis vector for $V_{\ical_f}$. Since $V_{\ical_f}^\perp$ is the column span of $\vand(A_{\ical_f})$, an $l$-dimensional vector will belong to $V_{\ical_f}$ if its components are the coefficients of a polynomial with $A_{\ical_f}$ as its roots. In particular, we consider the polynomial $\prod_{a_{k,m} \in A_{\ical_f}} (x - a_{k,m}) = \prod_{k \in \ical_f} \prod_{m=1}^{2\ei}(x - a_{k,m})$ which is of degree $2\ei |\ical_f| = 2\ei(K-r) = l-1$. Since the characteristic of the underlying field $\Fb_q$ is $2$, the coefficients of this polynomial are (starting from the smallest degree term)
\begin{equation*} 
~~\prod_{a_{k,m} \in A_{\ical_f}} a_{k,m}, ~~\sum_{S \in \text{{$\binom{A_{\ical_f}}{l-2}$}}}\prod_{a_{k,m} \in S} a_{k,m},~~\dots~~,~~\sum_{S \in \binom{A_{\ical_f}}{2}} \prod_{a_{k,m} \in S} a_{k,m},~~ \sum_{a_{k,m} \in A_{\ical_f}}\!\! a_{k,m},~~1.
\end{equation*} 
The $f^{\text{th}}$ column vector $\hb_f$ of the encoder matrix is the vector whose $l$ coordinates are these coefficients, that is, the $i^{\text{th}}$ component of $\hb_f$ is as follows
\begin{equation} \label{eq:components_of_H}
h_{i,f} = \sum_{S \in \binom{A_{\ical_f}}{l-i}} \prod_{a_{k,m} \in S} a_{k,m},
\end{equation}
with the convention that multiplication over an empty set of scalars is $1$, i.e., $h_{l,f}=1$.

We use this structure of $\hb_f$, along with the Schwartz-Zippel lemma, to arrive at this next result.

\begin{lemma} \label{lem:construction_any_2ei_independent}
Any set of $2\ei$ columns of $\Hb$ are linearly independent with probability $1 - O(q^{-1})$.
\end{lemma}
\begin{proof}
See Appendix~\ref{app:lem:construction_any_2ei_independent}.
\end{proof}

\subsubsection*{Proof of Theorem~\ref{thm:new_coding_scheme}}
We are now ready to prove the main result of this section by showing that the proposed random construction of $\Hb$ satisfies the criteria of Theorem~\ref{thrm1}. Consider any $k \in \kcal$. 
From~Lemma~\ref{lem:construction_any_2ei_independent}, we know that any $2\ei$ columns of $\Hb_{\xcal_k}$ are linearly independent with high probability.

Consider any subfile $f \in \ycal_k$. Since $p_{f,k} \neq \ast$, we have $k \in \ical_f$. Hence, $V_{\ical_f} \subset V_k$, and therefore, $\hb_f$, which is a basis vector of $V_{\ical_f}$, will be in $V_k$. 
Thus, we have 
\begin{equation} \label{eq:Hyk_subset_Vk}
\colspan(\Hb_{\ycal_k}) \subset V_k.
\end{equation}

By construction, the random subspaces $V_k$, $k \in \kcal$ are statistically independent.
Now consider any subfile $f \in \xcal_k$. Since $p_{f,k} = \ast$, $k$ does not belong to $\ical_f$. Hence, for every $f \in \xcal_k$, the subspace $V_{\ical_f}$ and the vector $\hb_f \in V_{\ical_f}$ are statistically independent of $V_k$. 
It follows that the subspace spanned by a given set of $2\ei$ columns of $\Hb_{\xcal_k}$ is statistically independent of $V_k$. Let $\mathcal{A}_k$ be any $2\ei$-sized subset of $\xcal_k$. Using~\eqref{eq:Hyk_subset_Vk}, Lemmas~\ref{lem:random_vandermonde} and~\ref{lem:construction_any_2ei_independent}, and the fact that $\Hb_{\mathcal{A}_k}$ and $V_k$ are statistically independent, we have
\begin{align*}
P\left[ \colspan(\Hb_{\mathcal{A}_k}) \cap \colspan(\Hb_{\ycal_k}) \neq \{\pmb{0}\}\right] &\leq P\left[ \colspan(\Hb_{\mathcal{A}_k}) \cap V_k \neq \{\pmb{0}\}\right] \\
&\leq P[\rank(\Hb_{\mathcal{A}_k}) \neq 2\ei] + P\left[ \colspan(\Hb_{\mathcal{A}_k}) \cap V_k \neq \{\pmb{0}\}~|~\rank(\Hb_{\mathcal{A}_k}) = 2\ei\right] \\
&= O(q^{-1}) + \sum_{\substack{U \subset \Fb_q^l \\ \dim(U)=2\ei}} P[\colspan(\Hb_{\mathcal{A}_k})=U] \, P[U \cap V_k \neq \{\pmb{0}\} \,| \,\colspan(\Hb_{\mathcal{A}_k})=U ] \\
&= O(q^{-1}) + \sum_{\substack{U \subset \Fb_q^l \\ \dim(U)=2\ei}} P[\colspan(\Hb_{\mathcal{A}_k})=U] \, P[U \cap V_k \neq \{\pmb{0}\}] \\
&= O(q^{-1}) + \sum_{\substack{U \subset \Fb_q^l \\ \dim(U)=2\ei}} P[\colspan(\Hb_{\mathcal{A}_k})=U] \, O(q^{-1}) \\
&= O(q^{-1}) + O(q^{-1}) P[\rank(\Hb_{\mathcal{A}_k})=2\ei] \\
&= O(q^{-1}).
\end{align*} 
Using a union bound argument, we immediately deduce that $\Hb$ satisfies all the design conditions of Theorem~\ref{thrm1} with probability $1-O(q^{-1})$.


\section{Converse Bounds} \label{sec:lower_bounds}

We derive lower bounds on the optimal communication cost $\lopt$ in this section. 
In Section~\ref{sub:sec:converse:general} we exhibit results applicable to all PDAs. 
We first show that $2\ei \leq \lopt \leq F$ and characterize the update problems with near-extreme values of $\lopt$, i.e., problems with $\lopt$ close to $2\ei$ and $F$.
We then derive a generic lower bound on $\lopt$ (Theorem~\ref{thm:generic_lower_bound}) that will be used in the rest of this section as a benchmark for our achievable schemes.
%
In Section~\ref{sub:sec:shangguan} we apply these results to a PDA designed by Shangguan et al.~\cite{SZG_IT_18}. In Sections~\ref{sub:sec:MN-PDA} and~\ref{sub:sec:uv_PDA} we consider the Maddah-Ali \& Niesen caching scheme~\cite{Ali_Niesen_2014} and a PDA independently designed by Yan et al.~\cite{YCTC_2017} and Tang \& Ramamoorthy~\cite{TaR_IT_18}, respectively. 
We show that for these latter two families of PDAs the update schemes proposed in this paper are optimal up to a constant multiplicative factor under some operating regimes. 
We also prove the following strong result for these two classes of PDAs when the updates are sufficiently sparse: as the number of nodes in the system increases while the caching ratio is kept constant, the cost $l = 2\ei(K-r)+1$ of the scheme of Theorem~\ref{thm:new_coding_scheme} satisfies $\frac{l}{\lopt} \to 1$ if $\log_2 \ei = o(\log_2 F)$.

\subsection{General Results for any PDA} \label{sub:sec:converse:general}

\subsubsection{Near-Extreme Communication Costs} 

We know from~\eqref{eq:trivial_upper_bound} that $\lopt$ is at the most $F$. We now identify the update problem scenarios for which $\lopt$ takes this largest possible value.

\begin{lemma} \label{lem:condition_lopt_F}
For the $(\xcal,\ei)$ cache update problem based on a $(K,F,Z,S)$ PDA, $\lopt = F$ if and only if $Z \leq 2\ei$.
\end{lemma}
\begin{proof}
If $Z \geq 2\ei + 1$, from Lemma~\ref{ub1}, $\lopt \leq F - (Z - 2\ei)^+ \leq F-1$. 
On the other hand, if $Z \leq 2\ei$, then $|\xcal_k|=Z \leq 2\ei$ for all $k \in \kcal$. Using Theorem~\ref{thm:lower_bound_XS} we see that $S=\kcal$, $\xcal_S = \fcal$ and $\lopt \geq |\xcal_S| + \min\{2\ei,F-|\xcal_S|\} = F - \min\{2\ei,0\} = F$.
\end{proof}

In other words, if the number of subfiles to be updated $\ei$ is $Z/2$ or more then broadcasting the updated file contents without coding is optimal with respect to communication cost.

Let us now assume a non-trivial coding scenario, i.e., $Z \geq 2\ei + 1$. Applying Theorem~\ref{thm:lower_bound_XS} to such a scenario, we see that $S=\phi$, $\xcal_S=\phi$, and hence, $\lopt \geq \min\{2\ei,F\}=2\ei$, where we have used $2\ei < Z \leq F$. Hence, $2\ei$ is a lower bound for $\lopt$ if $Z \geq 2\ei + 1$. The following result is useful in determining when the optimal communication cost $\lopt$ takes values close to $2\ei$.

\begin{lemma} \label{lem:H_2_columns}
For any $\ei \geq 1$ and any PDA, if $\Hb \in \Fb_q^{l \times F}$ is a valid encoder for a cache update problem then any two columns of $\Hb$ are linearly independent.
\end{lemma}
\begin{proof}
Let the columns of $\Hb$ be indexed by $\fcal$. We first prove that every column of $\Hb$ is non-zero. For any $f \in \fcal$, there must exist a $k \in \kcal$ such that $f \in \xcal_k$ since $r \geq 1$. 
From Corollary~\ref{corr1} we conclude that the column vector indexed by $f$ must be linearly independent by itself, i.e., must be non-zero. 

Consider two columns with indices $f,f' \in \fcal$. If $f,f' \in \xcal_k$ for some choice of $k$, then by Corollary~\ref{corr1} and using the fact $2\ei \geq 2$, these two column vectors are linearly independent. On the other hand, if $f,f'$ are such that $f \in \xcal_k$ and $f' \notin \xcal_k$ for some choice of $k$, then we observe that $f' \in \ycal_k$. The column indexed by $f'$ must be non-zero, and from Theorem~\ref{thrm1}, its span must not include the column indexed by $f$.
\end{proof}

We are now ready to identify the scenarios when $\lopt = 2\ei$.

\begin{lemma} \label{lem:cost_2ei}
Assume $Z \geq 2\ei + 1$. The optimal communication cost of updating the cache contents based on a $(K,F,Z,S)$ PDA satisfies $\lopt \geq 2\ei$ and attains equality if and only if $Z=F$.
\end{lemma}
\begin{proof}
We have already shown earlier in this section that $\lopt \geq 2\ei$ if $Z \geq 2\ei + 1$.

Now suppose $Z=F$. Using the MDS-based coding scheme and Lemma~\ref{ub1}, we have $\lopt \leq F - (Z - 2\ei) = 2\ei$.

On the other hand, consider the case $Z < F$. For any $k \in \kcal$, $|\ycal_k| \geq 1$. Now consider any $2\ei$ columns whose indices lie in $\xcal_k$ and any one column with index in $\ycal_k$. From Lemma~\ref{lem:H_2_columns}, the column with index in $\ycal_k$ is non-zero. Further, using Theorem~\ref{thrm1}, this column along with the $2\ei$ former columns form a linearly independent set. Hence, $\lopt \geq \rank(\Hb) \geq 2\ei + 1$.
\end{proof}

The scenario when $\lopt=2\ei$, or equivalently $Z=F$, corresponds to a trivial cache placement since every node caches all the subfiles. The next possible values of $\lopt$ are $2\ei+1$ and $2\ei+2$, which we consider next.

\begin{lemma} \label{lem:cost_2ei_plus}
Assume $Z \geq 2\ei+1$. 
The optimal communication cost $\lopt$ is
\begin{enumerate}
\item $\lopt = 2\ei + 1$ if $Z=F-1$, and
\item $\lopt = 2\ei + 2$ if $Z=F-2$.
\end{enumerate} 
\end{lemma}
\begin{proof}
Using the code design based on the parity-check matrix of MDS codes, we know that $\lopt \leq 2\ei+1$ and $\lopt \leq 2\ei + 2$ for $Z=F-1$ and $Z=F-2$, respectively.

To obtain a converse for $Z=F-1$, we use Lemma~\ref{lem:cost_2ei} which states that $\lopt=2\ei$ if and only if $Z=F$. Hence, $Z=F-1$ necessarily implies $\lopt > 2\ei$, i.e, $\lopt \geq 2\ei+1$.


We now prove the converse for $Z=F-2$.
Let $\Hb$ be any valid encoder matrix. For any $k \in \kcal$, $|\ycal_k| = 2$. By Lemma~\ref{lem:H_2_columns} these two columns are linearly independent, and by Theorem~\ref{thrm1} their column span intersects with the column span of any $2\ei$ columns from $\Hb_{\xcal_k}$ only at $\pmb{0}$. Since these $2\ei$ columns are linearly independent, we deduce that $\rank(\Hb) \geq 2\ei + 2$. Hence, $\lopt \geq 2\ei + 2$.
\end{proof}

\subsubsection{A generic lower bound} \label{sub:sec:generic_lower_bound}

The following lower bound is applicable to any PDA. In the sequel we will use this result to derive good lower bounds for specific well-known families of PDAs from the literature. Consider any choice of $L \in [K]$. Our lower bound is obtained by considering any $L$ out of the $K$ caching nodes in a sequence, say $\pi_1,\dots,\pi_L \in \kcal$, and counting the number of subfiles cached in each node $\pi_i$ which are not cached in any of the earlier nodes $\pi_1,\dots,\pi_{i-1}$ in the sequence. 

\begin{theorem} \label{thm:generic_lower_bound}
For any $(\xcal,\ei)$ update problem and any choice of $L \leq K$, let $\pi_1,\dots,\pi_L \in \kcal$ be indices of distinct nodes. Then
$\lopt \geq \sum_{i=1}^{L} \min \left\{\,2\ei, ~\left|\xcal_{\pi_i} \setminus \left(\xcal_{\pi_1} \cup \cdots \cup \xcal_{\pi_{i-1}}\right) \right|~\right\}$.
\end{theorem}
\begin{proof}
For each $i \in [L]$, let $\mathcal{A}_{\pi_i}$ be any subset of $\xcal_{\pi_i} \setminus \left(\xcal_{\pi_1} \cup \cdots \cup \xcal_{\pi_{i-1}}\right)$ of size $\min \{2\ei, \, |\xcal_{\pi_i} \setminus (\xcal_{\pi_1} \cup \cdots \cup \xcal_{\pi_{i-1}}) | \}$. Note that $\mathcal{A}_{\pi_1},\dots,\mathcal{A}_{\pi_L}$ are disjoint subsets of $\fcal$, and that the lower bound claimed in the theorem is $\sum_{i=1}^{L} |\mathcal{A}_{\pi_i}| = |\mathcal{A}_{\pi_1} \cup \cdots \cup \mathcal{A}_{\pi_L}|$.
We will show that if $\Hb$ is any valid encoder matrix for this update problem, then the columns of $\Hb$ indexed by $\mathcal{A}_{\pi_1} \cup \cdots \cup \mathcal{A}_{\pi_L}$ are linearly independent. Then, the number of transmissions required, which is equal to the number of rows of $\Hb$, is lower bounded by the rank of $\Hb$, which in turn is lower bounded by the size $|\mathcal{A}_{\pi_1} \cup \cdots \cup \mathcal{A}_{\pi_L}|$ of this set of linearly independent columns. 

Observe that $\mathcal{A}_{\pi_i} \subset \xcal_{\pi_i}$ and $|\mathcal{A}_{\pi_i}| \leq 2\ei$. Thus, from Corollary~\ref{corr1}, the columns of $\Hb$ indexed by $\mathcal{A}_{\pi_i}$ are linearly independent. For any $j > i$, we note that $\mathcal{A}_{\pi_j} \cap \xcal_{\pi_i} = \phi$, i.e., $\mathcal{A}_{\pi_j} \subset \ycal_{\pi_i}$. Considering all values of $j > i$, we then deduce that $\mathcal{A}_{\pi_{i+1}} \cup \cdots \cup \mathcal{A}_{\pi_L} \subset \ycal_{\pi_i}$.
Now applying Theorem~\ref{thrm1}, we see that the column span of $\Hb_{\mathcal{A}_{\pi_i}}$ intersects with the column span of $\Hb_{\mathcal{A}_{\pi_{i+1}} \cup \cdots \cup \mathcal{A}_{\pi_L}}$ only at $\pmb{0}$. 
In summary, the matrix $\Hb_{\mathcal{A}_{\pi_{1}} \cup \cdots \cup \mathcal{A}_{\pi_L}}$ can be partitioned into submatrices $\Hb_{\mathcal{A}_{\pi_1}},\dots,\Hb_{\mathcal{A}_{\pi_L}}$, and each submatrix $\Hb_{\mathcal{A}_{\pi_i}}$ has linearly independent columns and its column span intersects trivially with the column span of the submatrices appearing later in the sequence, i.e., with the column span of $\Hb_{\mathcal{A}_{\pi_{i+1}} \cup \cdots \cup \mathcal{A}_{\pi_L}}$. Hence, $\Hb_{\mathcal{A}_{\pi_{1}} \cup \cdots \cup \mathcal{A}_{\pi_L}}$ has linearly independent columns.
\end{proof}

\subsection{The Construction-I PDA of Shangguan et al.} \label{sub:sec:shangguan} 

In this sub-section we consider a class of PDAs given by Construction-I of Shangguan et~al.~\cite{SZG_IT_18} using hypergraphs and equivalently by Yan et al.~\cite{YTCC_COMML_18} using strong edge coloring of bipartite graphs. This family of PDAs includes the Maddah-Ali \& Niesen placement as a special case. We will derive a lower bound for this class of PDAs, and then specialize our bound to the Maddah-Ali \& Niesen PDA in Section~\ref{sub:sec:MN-PDA}.
This family of PDAs is characterized by three positive integers $n,a,b$ such that $a + b \leq n$. The subfiles are indexed by $a$-sized subsets of $[n]$, i.e., $\fcal=\binom{[n]}{a}$, and the users are indexed by $b$-sized subsets of $[n]$, $\kcal= \binom{[n]}{b}$. A subfile $f \in \binom{[n]}{a}$ is cached at user $k \in \binom{[n]}{b}$ if and only if $f \cap k \neq \phi$. This PDA has $F=\binom{n}{a}$, $K=\binom{n}{b}$, $Z = \binom{n}{a} - \binom{n-b}{a}$.

We will assume that $Z \geq 2\ei+1$, since otherwise we know that $\lopt=F$.
To derive a lower bound on the optimal communication cost of updating the cache contents, we use Theorem~\ref{thm:generic_lower_bound} with $L=n-b+1$ nodes indexed by
\begin{equation*}
\pi_1 = \{1,\dots,b\}, \, \pi_2 = \{2,\dots,b+1\},\,\dots,\, \pi_i=\{i,\dots,b+i-1\}, \, \dots, \, \pi_L = \{n-b+1,\dots,n\}.
\end{equation*} 
Observe that $|\xcal_{\pi_1}|=Z > 2\ei$, and for any $i=2,\dots,L$, $\xcal_{\pi_i} \setminus (\xcal_{\pi_1} \cup \cdots \cup \xcal_{\pi_{i-1}})$ is the collection of all $a$-sized subsets of $[n]$ that contain $b+i-1$ and do not contain any of $1,\dots,b+i-2$. Thus, $|\xcal_{\pi_i} \setminus (\xcal_{\pi_1} \cup \cdots \cup \xcal_{\pi_{i-1}})| = \binom{n-(b+i-1)}{a-1}$. From Theorem~\ref{thm:generic_lower_bound}, we have
\begin{align}
\lopt &\geq \min\{2\ei,|\xcal_{\pi_1}|\} + \sum_{i=2}^{n-b+1} \min\{2\ei, |\xcal_{\pi_i} \setminus (\xcal_{\pi_1} \cup \cdots \cup \xcal_{\pi_{i-1}})|\} \nonumber \\
&= 2\ei + \sum_{i=2}^{n-b+1} \min\left\{2\ei, \binom{n-b-i+1}{a-1}\right\} \nonumber \\
&= 2\ei + \sum_{j=0}^{n-b-1} \min\left\{2\ei, \binom{j}{a-1} \right\}. \label{eq:shangguan:1}
\end{align} 
Let $a_0$ be the smallest integer such that $\binom{a_0}{a-1} \geq 2\ei$. Then $\min\{2\ei,\binom{j}{a-1}\} = \binom{j}{a-1}$ if and only if $j \leq a_0-1$. Using this in~\eqref{eq:shangguan:1} we have, if $a_0 \leq n-b$ then
\begin{align}
\textstyle \lopt &\geq 2\ei + \sum_{j=0}^{a_0-1} \binom{j}{a-1} + \sum_{j=a_0}^{n-b-1} 2\ei 
= 2\ei + \binom{a_0}{a} + 2\ei(n-b-a_0) = 2\ei(n-b-a_0+1) + \binom{a_0}{a}. \label{eq:shangguan:2}
\end{align} 

\begin{example}
This example shows that the bound~\eqref{eq:shangguan:2} is tight. Consider $n=5$, $a=b=2$ and $\ei=1$. Then, $F=K=\binom{5}{2}=10$, $Z=\binom{5}{2} - \binom{3}{2}=7$, which is greater than $2\ei$, and the value of $a_0$ is $2$. Then the lower bound~\eqref{eq:shangguan:2} on $\lopt$ is $5$. The achievability scheme using the parity-check matrix of MDS codes has codelength $F-(Z-2\ei)^+=5$, which meets this lower bound. Hence, $\lopt=5$ for this update problem.
\end{example}

We now identify the optimal communication cost for the family of PDAs corresponding to $a=1$, i.e., $\fcal=[n]$, $\kcal=\binom{[n]}{b}$ and a subfile $f \in [n]$ is cached at node $k \in \binom{[n]}{b}$ if and only if $f \in k$. In this case $F=n$ and $Z=b$. Assuming $Z=b \geq 2\ei+1$ and applying~\eqref{eq:shangguan:1}, we have
\begin{equation*}
\lopt \geq 2\ei + \sum_{j=0}^{n-b-1} \min\left\{ 2\ei, \binom{j}{0} \right\} = 2\ei + n - b.
\end{equation*} 
From Lemma~\ref{ub1}, $\lopt \leq F - Z + 2\ei = 2\ei + n - b$. Hence, the optimal communication cost is $\lopt=2\ei+n-b$.

\subsection{The Maddah-Ali \& Niesen PDA} \label{sub:sec:MN-PDA}

The placement scheme of Maddah-Ali and Niesen, denoted as MN PDA, is the special case of the family considered in Section~\ref{sub:sec:shangguan} corresponding to $b=1$. 
We will follow the standard notation used in the literature, i.e., $\kcal=[K]$, $\fcal=\binom{[K]}{t}$ and a subfile $f \in \binom{[K]}{t}$ is cached at node $k \in [K]$ if and only if $k \in f$. Hence, $F=\binom{K}{t}$, $Z=\binom{K-1}{t-1}$ and $r=t$. Compared to the notation used in Section~\ref{sub:sec:shangguan}, $t$ and $K$ replace the symbols $a$ and $n$, respectively. 

Note that $\lopt=F$ when $Z = \binom{K-1}{t-1} \leq 2\ei$. In the rest of this sub-section we will assume $Z= \binom{K-1}{t-1} \geq 2\ei + 1$, where $\ei \geq 1$. This implies $t \geq 2$. Since $\binom{K-1}{t-1} > 2\ei$, the smallest integer $a_0$ such that $\binom{a_0}{t-1} \geq 2\ei$ satisfies $a_0 \leq K-1 = n-b$. Hence, the lower bound~\eqref{eq:shangguan:2} holds.

The \emph{caching ratio} of the MN PDA is $Z/F = t/K$, which is the fraction of the overall library cached at each node; this is denoted as $M/N$ in the literature, where $N$ is the number of files in the library and $M$ is the cache size at each node in terms of number of files. It is common to assume that as the number of nodes in the system varies, the cache size of individual nodes remain the same, and hence, the caching ratio remains constant. We will assume that the caching ratio is a constant $\beta$, $0 < \beta < 1$, and $t=\beta K$. 
The number of subfiles $F=\binom{K}{\beta K}=O(2^{K H_2(\beta)})$, where $H_2(\beta) = \beta \log_2 (1/\beta) + (1-\beta) \log_2 (1/(1-\beta))$ is the binary entropy function.

We will consider three different operating regimes for updating the MN PDA based on the sparsity level of the update, i.e., how $\ei$ varies with $F$. 
For each case we show that the communication cost of the achievability schemes proposed in this paper are either optimal or within a constant multiplicative gap from the optimal scheme.

\subsubsection{Regime 1, $\ei \leq t/2$} \label{regime1_mnpda}
In this regime the update scheme of Section~\ref{sec:coding_scheme} yields the optimal communication cost. Observe that the communication cost of this scheme is $2\ei(K-t) + 1$. To show that this cost is optimal, we use the lower bound~\eqref{eq:shangguan:2}. Here $a_0$ is the smallest integer such that $\binom{a_0}{t-1} \geq 2\ei$. Since $\binom{t-1}{t-1} = 1 < 2\ei$ and $\binom{t}{t-1}=t \geq 2\ei$, we conclude that $a_0=t$. Hence,~\eqref{eq:shangguan:2} implies $\lopt \geq 2\ei(K-t)+1$. Thus, we have $\lopt = 2\ei(K-t) + 1$.

\subsubsection{Regime 2, Sparse Update} \label{regime2_mnpda}
We assume that the number of subfiles to be updated grows exponentially with $K$, i.e., $\ei = 2^{\gamma K}$ for some constant $\gamma$. 
If $\gamma < H_2 (\beta)$, then $\frac{\ei}{F} \to 0$ as $K \to \infty$. 
We will show that if $\gamma$ is sufficiently small and $K$ large, the communication cost of the update scheme of Section~\ref{sec:coding_scheme} is within a constant multiplicative gap from the optimal cost. 

\begin{lemma} \label{lem:MN:sparse_update}
Let $\ei = 2^{\gamma K}$ and $t \geq 2$. If $t=\beta K$, 
the communication cost $l=2\ei(K-t)+1$ of the update scheme in Theorem~\ref{thm:new_coding_scheme} satisfies
\begin{equation} \label{eq:lem:MN:sparse_update:1}
 \frac{l}{\lopt} \leq \frac{1 - \beta + \frac{1}{2\ei K}}{\left[ 1 - \frac{2}{K} - 2^{1/(\beta K - 1)} \left(\beta - \frac{1}{K}\right) 2^{\gamma K/(\beta K - 1)}  \right]^+}.
\end{equation} 
\end{lemma}
\begin{proof}
See Appendix~\ref{app:lem:MN:sparse_update}.
\end{proof}

We now apply this bound to the scenario when the sparsity parameter $\gamma$ is small and the number of users in the system $K$ is large. 
If $\gamma < \beta\log_2(\frac{1}{\beta})$, then Lemma~\ref{lem:MN:sparse_update} implies that $\lim \sup_{K \to \infty} \frac{l}{\lopt} \leq \frac{1-\beta}{1 - \beta 2^{\gamma/\beta}}$. 

As another corollary to Lemma~\ref{lem:MN:sparse_update}, we observe that if $\ei=2^{o(K)}$ is sub-exponential in $K$, i.e., ${\log_2 \ei}\,/\,{K} \to 0$, or equivalently $\log_2 \ei \, / \, \log_2 F \to 0$, as $K \to \infty$, then $\big(2\ei(K-t)+1\big)\,/\,\lopt \to 1$ as $K \to \infty$. We state this observation as

\begin{lemma} \label{lem:MN:main:1}
For the MN PDA, if $t/K \in (0,1)$ is a constant and $\log_2 \ei = o (\log_2 F)$, then the cost $l=2\ei(K-t)+1$ of the update scheme from Theorem~\ref{thm:new_coding_scheme} satisfies
$\lim_{K \to \infty} \frac{l}{\lopt} = 1$.
\end{lemma}

\subsubsection{Regime 3, Dense Update} \label{regime3_mnpda}
In this case, we assume that $\ei$ is a constant fraction of the number of subfiles $F$, i.e., $\ei = \alpha F$ for some $0 < \alpha < 1/2$ (if $\alpha \geq 1/2$, then $2\ei \geq F \geq Z$, and hence $\lopt = F$). We know from Lemma~\ref{lem:cost_2ei}, that $\lopt \geq 2\ei = 2 \alpha F$.
Consider the MDS codes based update scheme of Lemma~\ref{ub1}, with $l = F - Z + 2\ei$. Since, $\beta = t/K = Z/F$, we observe that
\begin{equation*}
\frac{l}{\lopt} \leq \frac{F - Z + 2\ei}{2\ei} = \frac{1 - \beta + 2\alpha}{2\alpha}.
\end{equation*} 
Hence, when $\ei$ is a constant fraction of the number of subfiles, the communication cost promised by Lemma~\ref{ub1} is within a constant multiplicative factor of the optimal cost.

\subsection{The Yan et al. and Tang \& Ramamoorthy PDA} \label{sub:sec:uv_PDA}

We consider a family of PDAs constructed by Yan et al.~\cite{YCTC_2017} and Tang and Ramamoorthy~\cite{TaR_IT_18} that have smaller subpacketization than the MN PDA. This family of PDAs is parameterized by two integers $q,m \geq 2$.
Let $\Zb_q=\{0,1,\dots,q-1\}$ denote the additive cyclic group of order $q$.
The subfiles are indexed by $\fcal = \Zb_q^m$, the vectors of length $m$ over $\Zb_q$. The index set of nodes is $\kcal = \{(u,v)~|~ 1 \leq u \leq m+1, v \in \Zb_q\}$. 
For $1 \leq u \leq m$ and any $v \in \Zb_q$, the set $\xcal_{(u,v)}$ (which consists of the indices of the subfiles cached at the user $(u,v)$) contains all vectors from $\Zb_q^m$ whose $u^{\text{th}}$ coordinate is equal to $v$. For $u=m+1$ and any $v \in \Zb_q$, $\xcal_{(m+1,v)}$ contains a vector from $\Zb_q^m$ if and only if the sum of its coordinates over $\Zb_q$ is equal to $v$.
Note that $K=q(m+1)$, $F=q^m$, $Z=q^{m-1}$, $r=m+1$ and the caching ratio $\beta=Z/F=1/q$. 
The PDA in Fig.~\ref{fig:A32_PDA} is an instance of this family corresponding to parameters $q=3$ and $m=2$.

We will assume $Z \geq 2\ei + 1$. Consider the nodes $(u,v) \in \kcal$ with $1 \leq u \leq m$ and $v \in \Zb_q$. To apply Theorem~\ref{thm:generic_lower_bound}, we order these $qm$ nodes in the following sequence:
\begin{equation} \label{eq:uv_sequence}
(1,0),(2,0),\dots,(m,0),\,(1,1),(2,1),\dots,(m,1),\,\dots,(1,q-1),(2,q-1),\dots,(m,q-1).
\end{equation} 
A node $(u',v')$ appears earlier in the sequence than a node $(u,v)$ if either $v' < v$, or $v'=v$ and $u' < u$. In this case, we say that $(u',v')$ \emph{precedes} $(u,v)$ and denote this by $(u',v') \prec (u,v)$.
Let $x_{(u,v)}$ denote the number of subfiles in the node ${(u,v)}$ that are not contained in any of the preceding nodes, i.e.,
\begin{equation*}
 x_{(u,v)} = \Big| \xcal_{(u,v)} \, \setminus \bigcup_{\substack{(u',v'): \\ (u',v') \prec (u,v)}} \xcal_{(u',v')} \Big|.
\end{equation*} 
Applying Theorem~\ref{thm:generic_lower_bound} to this sequence of nodes, we have
\begin{equation} \label{eq:uv_PDA_general_lower_bound}
\lopt \geq \sum_{v=0}^{q-1} \sum_{u=1}^{m} \min \left\{ 2\ei, x_{(u,v)} \right\}.
\end{equation} 
The value of $x_{(u,v)}$ is the number of vectors $(s_1,\dots,s_m) \in \Zb_q^m$ that satisfy the following properties
\begin{enumerate}
\item $s_u=v$, 
\item if $u>1$, none of $s_1,\dots,s_{u-1}$ belong to $\{0,\dots,v\}$, i.e., $s_1,\dots,s_{u-1} \in \{v+1,\dots,q-1\}$, 
\item if $u<m$, none of $s_{u+1},\dots,s_m$ belong to $\{0,\dots,v-1\}$, i.e., $s_{u+1},\dots,s_m \in \{v,\dots,q-1\}$.
\end{enumerate} 
In the second condition above, the set $\{v+1,\dots,q-1\}$ is empty if $v = q-1$, in which case there is no suitable choice for the coordinates $s_1,\dots,s_{u-1}$. We conclude that
\begin{equation} \label{eq:x_uv}
x_{(u,v)} = (q-v-1)^{u-1}(q-v)^{m-u} \text{ for any } u \in [m], v \in \Zb_q,
\end{equation} 
where we treat $0^0$ to be equal to $1$.
Note that $(q-v-1)^{m-1} \leq x_{(u,v)} \leq (q-v)^{m-1}$ for any $u \in [m]$ and $v \in \Zb_q$.

\begin{lemma} \label{lem:x_uv_decreasing}
The sequence $x_{(u,v)}$ is a decreasing sequence, that is, $x_{(u',v')} \geq x_{(u,v)}$ if $(u',v') \prec (u,v)$.
\end{lemma}
\begin{proof}
Let $(u',v') \prec (u,v)$.
If $v'=v$, then necessarily $u' < u$. From~\eqref{eq:x_uv}, it clear that $x_{(u',v')} \geq x_{(u,v)}$. On the other hand, if $v' < v$, then 
\begin{equation*}
x_{(u',v')} \geq (q-v'-1)^{m-1} \geq (q-v)^{m-1} \geq x_{(u,v)}.
\end{equation*} 
\end{proof}

Based on the fact that $x_{(u,v)}$ is a decreasing sequence, we define $(u_0,v_0)$ to be the first index in the sequence~\eqref{eq:uv_sequence} such that $x_{(u_0,v_0)} < 2\ei$. A counting exercise leads us to

\begin{lemma} \label{lem:uv_PDA_good_lower_bound}
The optimal communication cost $\lopt \geq 2\ei\left(\,v_0(m-1) + u_0 + q-2\,\right)$, and $u_0 > 1$.
\end{lemma}
\begin{proof}
See Appendix~\ref{app:lem:uv_PDA_good_lower_bound}.
\end{proof}

To illustrate the goodness of our achievability schemes we consider two different operating regimes. We will assume that the caching ratio $Z/F = 1/q$ is a constant, and the number of nodes in the system is increased by increasing the value of $m$. Note that the subpacketization $F=q^m$ is exponential in $m$. 

\subsubsection{Regime 1, Sparse Update} \label{regime1_yanpda}

We assume that $\ei = \gamma^m$ for some constant $\gamma \in [1,q)$. In this case the fraction of subfiles being updated $\ei/F \to 0$ as $m \to \infty$. To apply Lemma~\ref{lem:uv_PDA_good_lower_bound}, we will first derive a lower bound on $v_0$. Since $u_0 > 1$ (from Lemma~\ref{lem:uv_PDA_good_lower_bound}) and $(u_0-1,v_0)$ immediately precedes $(u_0,v_0)$, we have
\begin{align*}
(q-v_0)^{m-1} \geq x_{(u_0-1,v_0)} \geq 2\ei > x_{(u_0,v_0)} \geq (q-v_0-1)^{m-1}.
\end{align*} 
This implies $q - (2\ei)^{\frac{1}{m-1}} - 1 < v_0 \leq q - (2\ei)^{\frac{1}{m-1}}$. Using this and $u_0 \geq 2$ in Lemma~\ref{lem:uv_PDA_good_lower_bound}, we have
\begin{align} \label{eq:uv_PDA_intemediate:1}
\lopt \geq 2\ei\left( (m-1)(q-2^{\frac{1}{m-1}}\gamma^{\frac{m}{m-1}}-1) + q \right).
\end{align} 
Comparing this with the communication cost $l=2\ei(K-r)+1=2\ei(m+1)(q-1) + 1$ guaranteed by Theorem~\ref{thm:new_coding_scheme}, we arrive at

\begin{lemma} \label{lem:uv_PDA_sparse:1}
Consider the Yan et al. and Tang \& Ramamoorthy PDA, with $Z/F=1/q$ being a constant and $\ei = \gamma^m$ for $\gamma \in [1,q-1)$. The communication cost of the scheme in Theorem~\ref{thm:new_coding_scheme} for updating this PDA satisfies $\limsup_{m \to \infty} \frac{l}{\lopt} \leq \frac{q-1}{q-1-\gamma}$.
\end{lemma}
\begin{proof}
From~\eqref{eq:uv_PDA_intemediate:1}, $\frac{l}{\lopt}$ is upper bounded by 
\begin{align*}
 \frac{(m+1)(q-1) + \frac{1}{2\ei}}{(m-1)(q-2^{\frac{1}{m-1}}\gamma^{\frac{m}{m-1}}-1) + q} =  
 \frac{\frac{m+1}{m-1}(q-1) + \frac{1}{2 \gamma^m (m-1)}}{q-2^{\frac{1}{m-1}}\gamma^{\frac{m}{m-1}}-1 + \frac{q}{m-1}}.  
\end{align*} 
The lemma follows from observing that $2^{\frac{1}{m-1}} \to 1$ and $\gamma^{\frac{m}{m-1}} \to \gamma$ as $m \to \infty$.
\end{proof}

It is clear from Lemma~\ref{lem:uv_PDA_sparse:1} that if $\ei$ is subexponential in $m$, i.e., if $\log_2 \ei = o(m)=o(\log_2 F)$, then $\lim_{m \to \infty} \frac{l}{\lopt}=1$.

We now consider the case $q=2$. For this case $x_{(u,v)}=2^{m-u}$ if $v=0$, $x_{(1,1)}=1$, and $x_{(u,v)}=0$ for $v=1$, $u>1$. 
Clearly, $v_0=0$ and $u_0$ satisfies $2^{m-u_0+1} \geq 2\ei > 2^{m-u_0}$. Hence, we have $m - \log_2 (2\ei) < u_0 \leq m - \log_2(2\ei) + 1$. Applying Lemma~\ref{lem:uv_PDA_good_lower_bound}, 
\begin{align*}
 \lopt \geq 2\ei u_0 \geq 2\ei(m - \log_2 (2\ei)) = 2\ei\left(\,m(1-\log_2 \gamma) - 1 \,  \right).
\end{align*} 
The ratio of $l=2\ei(K-r)+1=2\ei(m+1) + 1$ to this lower bound tends to $\frac{1}{(1 - \log_2 \gamma)}$ as $m \to \infty$. Hence, we have proved

\begin{lemma} \label{lem:uv_PDA_sparse:2}
Consider the Yan et al. and Tang \& Ramamoorthy PDA, with $Z/F=1/2$ being a constant and $\ei = \gamma^m$ for $\gamma \in [1,2)$. The communication cost of the scheme in Theorem~\ref{thm:new_coding_scheme} for updating this PDA satisfies $\limsup_{m \to \infty} \frac{l}{\lopt} \leq \frac{1}{1-\log_2 \gamma}$.
\end{lemma}

\subsubsection{Regime 2, Dense Update} \label{regime2_yanpda}

We now consider the case where the number of subfiles being updated are a constant fraction of $F$, say $\ei = \alpha F$ for some $0 < \alpha < \frac{1}{2q}$ (if $\alpha > \frac{1}{2q}$, then $2\ei = 2\alpha F \geq \frac{F}{q} = Z$, and hence, $\lopt = F$).
From Lemma~\ref{lem:uv_PDA_good_lower_bound}, $\lopt \geq 2\ei(u_0 + q-2) \geq 2\ei q=2\alpha F q$. Comparing with the cost $l = F - Z + 2\ei = F - F/q + 2\alpha F$ from the scheme of Lemma~\ref{ub1}, 
\begin{align*}
\frac{l}{\lopt} \leq \frac{1 - \frac{1}{q} + 2\alpha}{2\alpha q},
\end{align*} 
which is a constant independent of the scaling parameter $m$.

To conclude, we have shown that $l=\min\{2\ei(K-r)+1,F - (Z-2\ei)^+\}$ is order optimal in some sparse and as well as dense operating regimes.

\section{Conclusion and Discussion} \label{conclusions}

We formulated the problem of pushing updates blindly into the client nodes of a coded caching system. 
We designed a new coded transmission strategy for blind updates, and this scheme has near-optimal communication cost when the updates are sufficiently sparse for two well-known families of PDAs.
On the other hand the simple scheme of using the parity-check matrices of MDS codes is order-optimal when the updates are dense. 

We are yet to device efficient decoding strategies for our achievability schemes, and explicitly characterize the optimal cost $\lopt$ in general.
Another line of work is to consider a probabilistic scenario where the contents of the file being replaced and that of the new file arise from a known joint probability distribution, and characterize the information-theoretically optimal communication load in such a case.

\appendices

\section{Proof of Lemma~\ref{lem:construction_any_2ei_independent}} \label{app:lem:construction_any_2ei_independent}

We will rely on the Schwartz-Zippel lemma for our proof.

\begin{theorem}[{The Schwartz-Zippel lemma}] \label{thm:schwartz-zippel}
Let $f(t_1,\dots,t_n) \in \Fb_q[t_1,\dots,t_n]$ be a non-zero polynomial of total degree $d$ over $\Fb_q$. If $a_1,\dots,a_n$ are chosen at random independently and uniformly from $\Fb_q$, then $f(a_1,\dots,a_n)=0$ with probability at the most $d/q$.
\end{theorem}

In other words, if the multivariate polynomial $f$ is not the zero polynomial, and if $a_1,\dots,a_n$ are independent and uniformly distributed over $\Fb_q$, then the probability that $f(a_1,\dots,a_n) \neq 0$ is $1 - O(q^{-1})$.

In order to use the Schwartz-Zippel lemma, we will replace the random variables $a_{k,m}$ with indeterminates $t_{k,m}$, and consider the entries~\eqref{eq:components_of_H} of the encoder matrix as multivariate polynomials in  $t_{k,m}, k \in \kcal, m \in [2\ei]$. Analogous to $A_k$ and $A_{\ical_f}$, we also define $T_k = \{t_{k,m}~|~m \in [2\ei]\}$ and $T_{\ical_f} = \cup_{k \in \ical_f} T_k = \{t_{k,m}~|~k \in \ical_f, m \in [2\ei]\}$. 
Then, the $i^{\text{th}}$ component of the $f^{\text{th}}$ column of $\Hb$ is 
\begin{equation*}
 h_{i,f} = \sum_{S \in \binom{T_{\ical_f}}{l-i}} \prod_{t_{k,m} \in S} t_{k,m}.
\end{equation*} 
%
Now suppose $f_1,\dots,f_{2\ei}$ are distinct columns of $\Hb$. 
We will denote the $2\ei \times 2\ei$ submatrix of $[\hb_{f_1} \cdots \hb_{f_{2\ei}}]$ indexed by the rows 
$(i-1)(K-r)+1$, $i \in [2\ei]$,
as $\Bb=[b_{i,j}]$. Note that
\begin{equation*}
 b_{i,j} = \sum_{S \in \binom{T_{\ical_{f_j}}}{l-\left( (i-1)(K-r) + 1 \right)}} \prod_{t_{k,m} \in S} t_{k,m} = \sum_{S \in \binom{T_{\ical_{f_j}}}{(2\ei-i+1)(K-r)}} \prod_{t_{k,m} \in S} t_{k,m}, \text{ where } i,j \in [2\ei].
\end{equation*} 
We will show that the determinant of $\Bb$ is a non-zero polynomial. This will imply that $[\hb_{f_1} \cdots \hb_{f_{2\ei}}]$ has rank $2\ei$ with probability $1 - O(q^{-1})$ when the indeterminates $t_{k,m}$ are replaced by random elements from $\Fb_q$.

Denoting the set of all permutations on $[2\ei]$ by $\mathcal{P}$, and using the fact that the characteristic of $\Fb_q$ is $2$, we have
\begin{equation} \label{eq:detM_expansion}
\det(\Bb) = \sum_{\sigma \in \mathcal{P}} ~\prod_{j=1}^{2\ei} b_{\sigma(j),j} = \sum_{\sigma \in \mathcal{P}}~ \prod_{j=1}^{2\ei}~ \left( \sum_{\substack{S \subset T_{\ical_{f_j}}\\ |S|=(2\ei - \sigma(j)+1)(K-r)}}  \!\!\!\!\prod_{t_{k,m} \in S} t_{k,m} \right).
\end{equation} 
Considering $\det(\Bb)$ as a polynomial, each monomial in its expansion is of the form 
\begin{equation} \label{eq:each_monomial}
g(\sigma,S_1,\dots,S_{2\ei}) \triangleq \prod_{j=1}^{2\ei} \prod_{t_{k,m} \in S_j} t_{k,m}
\end{equation} 
for some choice of permutation $\sigma$ and subsets $S_1,\dots,S_{2\ei}$ where $S_j \subset T_{\ical_{f_j}}$ and $|S_j|=(2\ei-\sigma(j)+1)(K-r)$.
%
Note that $T_{\ical_{f_j}}=\{t_{k,m}~|~k \in \ical_{f_j}, m \in [2\ei]\}$.
One of the monomials in the expansion~\eqref{eq:detM_expansion}, corresponding to $\sigma$ being the identity permutation and 
$S_j = \{t_{k,m}~|~k \in \ical_{f_j}, j \leq m \leq 2\ei \}$ for each $j \in [2\ei]$,
is $g^* = \prod_{j=1}^{2\ei} \prod_{k \in \ical_{f_j}} \prod_{m=j}^{2\ei} t_{k, \, m}$.

We will show that $\det(\Bb)$ is a non-zero polynomial by proving that the monomial $g^*$ occurs exactly once in the expansion~\eqref{eq:detM_expansion}.
Now suppose that a permutation $\sigma$ and subsets $S_1,\dots,S_{2\ei}$ satisfying $S_j \subset T_{\ical_{f_j}}$ and $|S_j| = (2\ei - \sigma(j) + 1)(K-r)$ are such that the monomial $g(\sigma,S_1,\dots,S_{2\ei})$ in~\eqref{eq:each_monomial} equals $g^*$, i.e.,
\begin{equation} \label{eq:app:g_g}
 \prod_{j=1}^{2\ei} \prod_{t_{k,m} \in S_j} t_{k,m} = \prod_{j=1}^{2\ei} \prod_{k \in \ical_{f_{j}}} \prod_{m=j}^{2\ei} t_{k, \, m}.
\end{equation} 
We prove by induction on $j$ that, necessarily, $\sigma(j)=j$, $S_j = \{t_{k,\,m}~|~ k \in \ical_{f_j}, m \geq j\}$ for all $j \in [2\ei]$.
From~\eqref{eq:mild_condition_pda} we know that $\ical_{f_1},\dots,\ical_{f_{2\ei}}$ are distinct. This will be used in our proof.

Consider $j_0 = \sigma^{-1}(1)$. Note that $|S_{j_0}|=2\ei(K-r)$ and $S_{j_0} \subset T_{\ical_{f_{j_0}}}$, that is, $S_{j_0} = T_{\ical_{f_{j_0}}}$. Thus, $\prod_{t_{k,m} \in T_{\ical_{f_{j_0}}}} t_{k,m} = \prod_{k \in \ical_{f_{j_0}}} \prod_{m=1}^{2\ei} t_{k, \, m}$ is a factor of $g(\sigma,S_1,\dots,S_{2\ei})$, and hence, a factor of $g^*$.
In particular, $\prod_{k \in \ical_{f_{j_0}}} t_{k,1}$ is a factor of the RHS of~\eqref{eq:app:g_g}.
Observe that the factors in the RHS of~\eqref{eq:app:g_g} of the form $t_{k,1}$, for some choice of $k$, are $\{t_{k,1}~|~k \in \ical_{f_1}\}$.
Since $\ical_{f_1},\dots,\ical_{f_{2\ei}}$ are distinct, we conclude that $\ical_{f_{j_0}}=\ical_{f_1}$. Hence $j_0=1$, and $\sigma(1)=1$.

Now let $i_0 \geq 2$. Assume $\sigma(j)=j$ and $S_j = \{t_{k,m}~|~k \in \ical_{f_j}, m \geq j \}$ for all $j \leq i_0-1$. From~\eqref{eq:app:g_g}, 
\begin{equation} \label{eq:app:g_g_ind}
g^*_{i_0} \triangleq \frac{g^*}{ \prod_{j=1}^{i_0-1} \prod_{t_{k,m} \in S_j} t_{k,m}} = \prod_{j=i_0}^{2\ei} \prod_{t_{k,m} \in S_j} t_{k,m} = \prod_{j=i_0}^{2\ei} \prod_{k \in \ical_{f_j}} \prod_{m=j}^{2\ei} t_{k,\,m}.
\end{equation} 
Let $j_0 = \sigma^{-1}(i_0)$. By the induction hypothesis, we know that $j_0 \geq i_0$. 
From the RHS of~\eqref{eq:app:g_g_ind}, we observe that each factor $t_{k,m}$ of $g^*_{i_0}$ satisfies $m \geq i_0$. Thus, for any $j \geq i_0$, we have $S_j \subset \{t_{k,\,m}~|~k \in \ical_{f_{j}}, m \geq i_0\}$.
Since $\sigma(j_0)=i_0$, we have $|S_{j_0}| = (K-r)(2\ei-i_0+1)$, and hence, we deduce that $S_{j_0} = \{t_{k,\,m}~|~k \in \ical_{f_{j_0}}, m \geq i_0\}$.
Using this with the fact that $\prod_{t_{k,m} \in S_{j_0}} t_{k,m}$ is a factor of $g^*_{i_0}$, we see that $\prod_{k \in \ical_{f_{j_0}}} t_{k,i_0}$ is a factor of $g^*_{i_0}$.
We also observe from the RHS of~\eqref{eq:app:g_g_ind} that the only factors of $g^*_{i_0}$ of the form $t_{k,i_0}$, for some choice of $k$, are $\{t_{k,i_0}~|~k \in \ical_{f_{i_0}}\}$.
Using the fact that $\ical_{f_1},\dots,\ical_{f_{2\ei}}$ are distinct, we conclude that $\ical_{f_{j_0}} = \ical_{f_{i_0}}$ i.e., $i_0=j_0$. Hence, $\sigma(i_0)=i_0$.
This completes the proof of the induction step. 

We conclude that the monomial $g^*$ appears exactly once in the expansion~\eqref{eq:detM_expansion} of $\det(\Bb)$, and hence this monomial does not vanish, and therefore, $\det(\Bb)$ is a non-zero polynomial in the indeterminates $t_{k,m}$, $k \in \kcal, m \in [2\ei]$.
From the Schwartz-Zippel lemma, we deduce that if $a_{k,m}$, $k \in \kcal, m \in [2\ei]$, are chosen independently and uniformly at random from $\Fb_q$, then the probability that $\det(\Bb)$ evaluates to a non-zero value at the evaluation point $(a_{k,m}~|~k \in \kcal, m \in [2\ei])$ is $1 - O(q^{-1})$. 
Since $\Bb$ is a $2\ei \times 2\ei$ submatrix of $[\hb_{f_1}~\cdots~\hb_{f_{2\ei}}]$, we conclude that $\hb_{f_1},\dots,\hb_{f_{2\ei}}$ are linearly independent with probability $1 - O(q^{-1})$.

\section{Proof of Lemma~\ref{lem:MN:sparse_update}} \label{app:lem:MN:sparse_update}

From~\eqref{eq:shangguan:2}, we know that if $a_0$ is the smallest integer such that $\binom{a_0}{t-1} \geq 2\ei$, then $\lopt$ is lower bounded by $2\ei(K-a_0) + \binom{a_0}{t}$. Using this, along with the result $\binom{a_0}{t} = \binom{a_0}{t-1} \times \frac{a_0 - t + 1}{t}$, we have
\begin{align*}
\lopt &\geq 2\ei(K-a_0) + \binom{a_0}{t-1}\, \frac{a_0-t+1}{t} \\
&\geq 2\ei(K-a_0) + 2\ei\frac{(a_0-t+1)}{t} \\
&= 2\ei\left( K - \frac{t-1}{t} - \frac{a_0(t-1)}{t} \right) \\
&\geq 2\ei(K-1-a_0), ~~~~~~~~~~~~\text{since } \frac{t-1}{t} \leq 1.
\end{align*} 
If $a'$ is any integer such that $\binom{a'}{t-1} \geq 2\ei$, then $a' \geq a_0$, and  
$\lopt \geq 2\ei(K-1-a_0) \geq 2\ei(K-1-a')$.
We now observe that $a' = \left\lceil (t-1) \times (2\ei)^{{1}/{(t-1)}} \right\rceil$ satisfies this condition, since
\begin{align*}
\binom{a'}{t-1} \geq \left( \frac{a'}{t-1} \right)^{t-1} \geq 2\ei.
\end{align*} 
Hence we have the lower bound
\begin{align*}
\lopt &\geq 2\ei(K-1-a') \geq 2\ei(K-1-\left\lceil (t-1) \times (2\ei)^{{1}/{(t-1)}} \right\rceil) \\ 
&\geq 2\ei\left(K-\,2\,-\,(t-1)(2\ei)^{{1}/{(t-1)}}\right) \\
&\geq 2\ei\left(K- 2 - (\beta K-1) 2^{1/(\beta K -1)} 2^{\gamma K / (\beta K - 1)} \right)
\end{align*} 
Using this with the trivial bound $\lopt \geq 0$, we have
\begin{align*}
&
\lopt \geq 2\ei K \left[ 1 - \frac{2}{K} - 2^{1/(\beta K - 1)} \left(\beta - \frac{1}{K}\right) 2^{\gamma K/(\beta K - 1)}  \right]^+.
\end{align*} 
Finally, comparing this lower bound with $l=2\ei(K-t)+1=2\ei K \left(1 - \beta + \frac{1}{2\ei K}\right)$, we arrive at the statement of this lemma.

\section{Proof of Lemma~\ref{lem:uv_PDA_good_lower_bound}} \label{app:lem:uv_PDA_good_lower_bound}


The sequence $x_{(u,v)}$ has the following properties. For any $v \in \Zb_q$,
\begin{align*}
\sum_{u=1}^{m} x_{(u,v)} = \sum_{u=1}^{m} (q-v-1)^{u-1}(q-v)^{m-u} = (q-v)^m - (q-v-1)^m.
\end{align*} 
This implies that 
$\sum_{v=v_0+1}^{q-1} \sum_{u=1}^{m}x_{(u,v)} = \sum_{v=v_0+1}^{q-1} (q-v)^m - (q-v-1)^m = (q-v_0-1)^m$.
Also, 
it is straightforward to see that
\begin{align*}
\sum_{u=u_0}^{m} x_{(u_0,v_0)} &= \sum_{u=u_0}^{m} (q-v_0-1)^{u-1}(q-v_0)^{m-u} \\
&= \frac{(q-v_0)^m}{(q-v_0-1)} \sum_{u=u_0}^{m} \left( \frac{q-v_0-1}{q-v_0} \right)^u \\
&= \frac{(q-v_0)^m}{(q-v_0-1)} \times \frac{\left( \frac{q-v_0-1}{q-v_0} \right)^{u_0} - \left( \frac{q-v_0-1}{q-v_0} \right)^{m+1}}{1 - \frac{q-v_0-1}{q-v_0}} \\
&=(q-v_0-1)^{u_0-1}(q-v_0)^{m-u_0+1} - (q-v_0-1)^m. 
%
\end{align*} 

We now argue that \mbox{$u_0 \geq 2$}. Let us assume the contrary, i.e., \mbox{$u_0=1$}. Then \mbox{$v_0 \geq 1$}, since otherwise $2\ei > x_{(u_0,v_0)}=x_{(1,0)} = Z$, which contradicts our assumption that $Z \geq 2\ei + 1$. Now notice that $(m,v_0-1) \prec (1,v_0)$ and $x_{(m,v_0-1)} = x_{(1,v_0)}$ since both are equal to $(q-v_0)^{m-1}$. Hence, $x_{(m,v_0-1)} < 2\ei$. This contradicts the assumption that $(u_0,v_0)$ is the first element in the sequence~\eqref{eq:uv_sequence} with value strictly less than $2\ei$ since $(m,v_0-1) \prec (1,v_0) = (u_0,v_0)$. Hence, we conclude that $u_0 \geq 2$.


The element immediately preceding $(u_0,v_0)$ is $(u_0-1,v_0)$ and 
$x_{(u_0-1,v_0)} \geq 2\ei$, that is, \mbox{$(q-v_0-1)^{u_0-2}(q-v_0)^{m-u_0+1} \geq 2\ei$}. Hence,
\begin{equation*}
\sum_{u=u_0}^{m} x_{(u_0,v_0)} = (q-v_0-1)^{u_0-1}(q-v_0)^{m-u_0+1} - (q-v_0-1)^m \geq 2\ei(q-v_0-1) - (q-v_0-1)^m.
\end{equation*} 

We now apply all the results derived in this appendix together with Lemma~\ref{lem:x_uv_decreasing} to obtain
\begin{align*}
\lopt \geq \sum_{(u,v)} \min\{2\ei,x_{(u,v)}\} &= \sum_{(u,v) \prec (u_0,v_0)} \!\!\!\!\!\!\! 2\ei ~+ \sum_{u=u_0}^{m} x_{(u_0,v_0)} + \sum_{v=v_0+1}^{q-1}\sum_{u=1}^{m} x_{(u,v)} \\
&\geq 2\ei(v_0m + u_0 - 1) + 2\ei(q-v_0-1) - (q-v_0-1)^m + (q-v_0-1)^m \\
&= 2\ei\left( v_0(m-1) + u_0 + q-2 \right).
\end{align*} 



\end{document}